\documentclass[12pt]{article}

\usepackage[breaklinks,pdfpagelabels]{hyperref} 

\usepackage{enumerate}
\usepackage{fullpage}
\usepackage{color}
\usepackage{xspace}
\usepackage{euscript}
\usepackage{amstext}
\usepackage{paralist}
\usepackage{amsmath}
\usepackage{amssymb}
\usepackage{graphicx}%

\usepackage{wide}%
\usepackage{lmodern}%

\usepackage{ntheorem} \theoremseparator{.}

\newtheorem{theorem}{Theorem}[section]

\newtheorem{lemma}[theorem]{Lemma}

{\theorembodyfont{\rm} \newtheorem{definition}[theorem]{Definition}}
{\theorembodyfont{\rm} }

\newcommand{\myparagraph}[1]{\paragraph{#1.}}

\newtheorem{observation}[theorem]{Observation}

\newcommand{\myqedsymbol}{\rule{2mm}{2mm}}
\newenvironment{proof}{\trivlist\item[]\emph{Proof}:}%
               {\unskip\nobreak\hskip 1em plus 1fil\nobreak%
                        \myqedsymbol%$\Box$
                        \parfillskip=0pt%
                        \endtrivlist}

\newcommand{\Do}{{\small\bf do}\ }%
\newcommand{\Return}{{\small\bf return\ }}%
\newcommand{\If}{{\small\bf if}\ }%
\newcommand{\Then}{{\small\bf then}\ }%
\newcommand{\Else}{{\small\bf else}\ }%

\newcommand{\For}{{\small\bf for}\ }

\newcommand{\xbeginlgox}{\begin{minipage}{1in}\begin{tabbing}
           \quad\=\qquad\=\qquad\=\qquad\=\qquad\=\qquad\=\qquad\=\kill}
        \newcommand{\xendlgox}{\end{tabbing}\end{minipage}}
\newenvironment{algorithm}{\begin{tabular}{|l|}\hline\xbeginlgox}
    {\xendlgox\\\hline\end{tabular}} 
\DefineNamedColor{named}{RedViolet} {cmyk}{0.07,0.90,0,0.34}
\newcommand{\AlgorithmI}[1]{{\textcolor[named]{RedViolet}{\texttt{\bf{#1}}}}}
\newcommand{\Algorithm}[1]{{\AlgorithmI{#1}\index{algorithm!#1@{\AlgorithmI{#1}}}}}
\providecommand{\Comment}[1]{\textcolor[named]{RedViolet}{\texttt{{\bf //} #1}}}

\providecommand{\SarielWWWPapersAddr}
{http://www.illinois.edu/\string~sariel/papers}

\newcommand{\SarielThanks}[1]{%
   \thanks{Department of Computer Science; University of Illinois; 201
      N. Goodwin Avenue; Urbana, IL, 61801, USA; {\tt
         \url{http://www.illinois.edu/\string~sariel/}.}%
      #1}}

\newcommand{\NirmanThanks}[1]{\thanks{Department of Computer Science;
      University of Illinois; 201 N. Goodwin Avenue; Urbana, IL,
      61801, USA; {\tt \si{nkumar5}\atgen{}illinois.edu}; {\tt
         \url{http://www.cs.illinois.edu/\string~\si{nkumar5}/}.}
      #1}} \newcommand{\atgen}{\symbol{'100}}

\newtheorem{remark}[theorem]{Remark}

\definecolor{blue25}{rgb}{0,0,0.65}
\newcommand{\emphic}[2]{%
     \textcolor{blue25}{%
         \textbf{\emph{#1}}}%
         \index{#2}}
\newcommand{\emphi}[1]{\emphic{#1}{#1}}

\providecommand{\lemlab}[1]{\label{lemma:#1}}
\providecommand{\lemref}[1]{Lemma~\ref{lemma:#1}}

\newcommand{\obslab}[1]{\label{observation:#1}}
\newcommand{\obsref}[1]{Observation~\ref{observation:#1}}

\providecommand{\lempntlab}[1]{\label{lemmapnt:#1}}
\providecommand{\lempntref}[1]{(\ref{lemmapnt:#1})}
\newcommand{\figlab}[1]{\label{fig:#1}}
\newcommand{\figref}[1]{Figure~\ref{fig:#1}}

\newcommand{\seclab}[1]{\label{sec:#1}}
\newcommand{\secref}[1]{Section~\ref{sec:#1}}

\newcommand{\subseclab}[1]{\label{subsec:#1}}
\newcommand{\subsecref}[1]{Subsection~\ref{subsec:#1}}
\newcommand{\remlab}[1]{\label{rem:#1}}
\newcommand{\remref}[1]{Remark~\ref{rem:#1}}
\newcommand{\thmlab}[1]{{\label{theo:#1}}}
\newcommand{\thmref}[1]{Theorem~\ref{theo:#1}}

\newcommand{\eqnlab}[1]{\label{equation:#1}}
\newcommand{\eqnref}[1]{(\ref{equation:#1})}
\newcommand{\eqlab}[1]{\label{equation:#1}}
\newcommand{\Eqref}[1]{Eq.~(\ref{equation:#1})}

\renewcommand{\Re}{{\rm I\!\hspace{-0.025em} R}}
\newcommand{\pth}[2][\!]{#1\left({#2}\right)}
\newcommand{\pbrc}[2][\!\!]{#1\left[ {#2} \MakesBig \right]}

\newcommand{\diameterX}[1]{\mathrm{d{}i{}am}\pth{#1}}
\newcommand{\diameterY}[2]{\mathrm{d{}i{}am}_{#1}\pth{#2}}

\newcommand{\Term}[1]{\textsf{#1}}
\newcommand{\TermI}[1]{\Term{#1}\index{#1@\Term{#1}}}

\newcommand{\aftermathA}{\par\vspace{-\baselineskip}}
\newcommand{\ceiling}[1]{\left \lceil {#1} \right \rceil}

\newcommand{\normX}[2]{\left\| {#1} - {#2}  \right\|}
\newcommand{\abs}[1]{\left | {#1} \right |}
\newcommand{\nnA}{\mathsf{n}_{\query}}
\newcommand{\algonnA}{y}

\newcommand{\distC}{\overline{r_1}}
\newcommand{\annAB}{\overline{\annA\!}}
\newcommand{\annBB}{\overline{\algoonnB\!}\,}
\newcommand{\closest}[1]{\alpha( {#1} )}
\newcommand{\MakeBig}{\rule[-.2cm]{0cm}{0.4cm}}
\newcommand{\MakesBig}{\rule[0.0cm]{0.0cm}{0.3cm}} % really small
\newcommand{\MakeSBig}{\rule[0.0cm]{0.0cm}{0.35cm}} % really small
\newcommand{\brc}[1]{\left\{ {#1} \right\}}
\newcommand{\sep}[1]{\,\left|\, {#1} \MakeBig\right.}

\newcommand{\centerset}{\mathcal{C}}
\newcommand{\centernettree}{\mathcal{N}_{\mathcal{C}}}

\newcommand{\dd}{\tau}
\newcommand{\ballC}{\mathsf{b}}
\newcommand{\ball}[2]{\mathsf{ball}\pth{ {#1}, {#2} }}
\newcommand{\ballExt}[3]{\mathsf{ball}_{#1}\pth{ {#2}, {#3} }}
\newcommand{\distM}[3]{\mathsf{d}_{{#1}}\pth{{#2},{#3}}}
\newcommand{\dist}[3][\!]{\mathsf{d}\pth[#1]{{#2},{#3}}}
\newcommand{\order}[2][\!]{O\pth[#1]{ {#2} }}
\newcommand{\aproxOrder}[1]{\widetilde{O}\pth{#1}}
\newcommand{\ordereq}[1]{\Theta \left ( {#1} \right )}
\newcommand{\ordergeq}[1]{\Omega \left ( {#1} \right )}

\providecommand{\si}[1]{#1}

\newcommand{\eps}{{\varepsilon}}%
\newcommand{\annA}{y_1}
\newcommand{\annAf}{f_1}
\newcommand{\algoonnB}{y_2}
\newcommand{\distA}{r}
\newcommand{\distB}{\rho_1}
\newcommand{\query}{\mathtt{q}}
\newcommand{\RegionA}{C}
\newcommand{\WSPD}{\TermI{WSPD}\xspace}
\newcommand{\ANN}{\TermI{ANN}\xspace}
\newcommand{\JL}{\TermI{JL}\xspace}
\newcommand{\AVD}{\TermI{AVD}\xspace}
\newcommand{\etal}{\textit{et~al.}\xspace}
\renewcommand{\Re}{{\rm I\!\hspace{-0.025em} R}}
\newcommand{\diameter}[1]{\mathsf{diam}\pth{ {#1} }}

\newcommand{\porder}[1]{O\pth[]{ {#1} }}

\newcommand{\nfrac}[2]{{#1}/{#2}}
\newcommand{\mapped}[1]{#1'}
\newcommand{\PntSet}{\mathsf{P}}
\newcommand{\PntSetB}{\mathsf{Q}}
\newcommand{\PntSetC}{\mathsf{R}}
\newcommand{\queryB}{\overline{\query}}
\newcommand{\nnB}{\mathsf{n}_{\queryB}}
\providecommand{\remove}[1]{}
\newcommand{\WSPDRep}{\mathcal{W}}
\newcommand{\pDist}{t}
\newcommand{\PDist}{T}
\newcommand{\lDist}[1]{\mathsf{l}_{#1}}
\newcommand{\LDist}[1]{\mathsf{L}_{#1}}
\newcommand{\cardin}[1]{\left\lvert {#1} \right\rvert}
\newcommand{\NNbrC}{\mathsf{n{}n}}

\newcommand{\NNbrInSet}[2]{\NNbrC\pth{#1,#2}}
\newcommand{\algANN}{\Algorithm{\si{algANN}}\xspace}
\newcommand{\algBuildANN}{\Algorithm{\si{algBuildANN}}}
\newcommand{\algNet}{\Algorithm{\si{compNet}}}
\newcommand{\algBuildAVD}{\Algorithm{\si{algBuildAVD}}}
\newcommand{\height}[1]{\mathsf{h}\pth{ #1 }}
\newcommand{\hMax}[1]{\mathsf{h}_{\max} \pth{#1}}
\newcommand{\brad}[1]{R\pth{ #1 }}
\newcommand{\RegionSet}{\EuScript{R}}
\newcommand{\QuerySetA}{\EuScript{U}}
\newcommand{\QuerySetB}{\EuScript{V}}
\newcommand{\pntA}{\mathsf{p}}
\newcommand{\pntB}{v}
\newcommand{\pntC}{u}
\newcommand{\pntD}{z}
\newcommand{\pntE}{w}
\newcommand{\pntQ}{q}
\newcommand{\diam}{\mathsf{D}}
\DeclareMathAlphabet{\mathpzc}{OT1}{pzc}{m}{it}
\newcommand{\constA}{\mathpzc{c}_1}
\newcommand{\constB}{\mathpzc{c}_2}
\newcommand{\Net}{N}
\newcommand{\mtrA}{\mathcal{X}}
\newcommand{\manifold}{\mathcal{M}}
\newcommand{\mtrB}{\manifold'}
\newcommand{\cPnt}{\overline{\mathrm{c}}}
\newcommand{\Tree}{\EuScript{T}}

\newcommand{\reachX}[1]{\mathsf{s}\pth{#1}}
\newcommand{\DS}{\mathcal{D}}
\newcommand{\nettree}{T}
\newcommand{\vertexA}{v}
\newcommand{\vertexB}{u}
\newcommand{\lvl}[2][\!]{l\pth[#1]{#2}}
\newcommand{\parent}[2][\!]{\overline{p}\pth[#1]{#2}}
\newcommand{\nettreeconst}{\gamma}
\newcommand{\nettreerep}[1]{\mathrm{rep}_{#1}}
\newcommand{\Integers}{\mathbb{Z}}

\newcommand{\setA}{\mathsf{B}}
\newcommand{\num}{\psi}

\newcommand{\sA}{{A}}%
\newcommand{\sB}{{B}}%

\newcommand{\pa}{{a}}%
\newcommand{\pb}{{b}}%
\newcommand{\modA}{{A}^+}
\begin{document}

\title{Approximate Nearest Neighbor Search for Low Dimensional
   Queries%
   \thanks{Work on this paper was partially supported by a NSF AF
      award \si{CCF}-0915984. A preliminary version of this paper
      appeared in SODA 2011 \cite{hk-annsl-11}.}%
}%

\author{%
   Sariel Har-Peled%
   \SarielThanks{}%
   \and%
   Nirman Kumar%
   \NirmanThanks{}}

\date{\today}

\maketitle

\begin{abstract}
    We study the Approximate Nearest Neighbor problem for metric
    spaces where the query points are constrained to lie on a subspace
    of low doubling dimension, while the data is high-dimensional. We
    show that this problem can be solved efficiently despite the high
    dimensionality of the data.
\end{abstract}

\section{Introduction}
The nearest neighbor problem is the following. Given a set $\PntSet$
of $n$ data points in a metric space $\mtrA$, preprocess $\PntSet$,
such that given a query point $\query \in \mtrA$, one can find
(quickly) the point $\nnA \in \PntSet$ closest to $\query$. Nearest
neighbor search is a fundamental task used in numerous domains
including machine learning, clustering, document retrieval, databases,
statistics, and many others.

\myparagraph{Exact nearest neighbor}
The (exact) nearest neighbor problem has a naive linear time algorithm
without any preprocessing. However, by doing some nontrivial
preprocessing, one can achieve a sublinear search time for the nearest
neighbor. In $d$-dimensional Euclidean space (i.e., $\Re^d$) this is
facilitated by answering point location queries using a Voronoi
diagram \cite{bcko-cgaa-08}.  However, this approach is only suitable
for low dimensions, as the complexity of the Voronoi diagram is
$\ordereq{n^{\ceiling{d/2}}}$ in the worst case.  Specifically,
Clarkson \cite{c-racpq-88} showed a data-structure with query time
$\order{\log n }$ time, and $\order{n^{\ceiling{d/2} + \delta}}$
space, where $\delta > 0$ is a prespecified constant (the $O(\cdot)$
notation here hides constants that are exponential in the
dimension). One can tradeoff the space used and the query time
\cite{am-rsps-93}. Meiser \cite{m-plah-93} provided a data-structure
with query time $\order{d^5 \log n }$, which has polynomial dependency
on the dimension, where the space used is $\order{ n^{ d +
      \delta}}$. These solutions are impractical even for data-sets of
moderate size if the dimension is larger than two.

\myparagraph{Approximate nearest neighbor}
In typical applications, it is usually sufficient to return an
\emphi{approximate nearest neighbor} (\emphi{\ANN{}}). Given an $\eps
> 0$, a $(1 + \eps)$-\ANN to a query point $\query$, is a point
$\algonnA \in \PntSet$, such that
\begin{equation*}
    \dist{\query}{\algonnA} \leq (1 + \eps)\dist{\query}{\nnA},
\end{equation*}
where $\nnA \in \PntSet$ is the nearest neighbor to $\query$ in
$\PntSet$. Considerable amount of work was done on this problem, see
\cite{c-nnsms-06} and references therein.

In high dimensional Euclidean space, Indyk and Motwani showed that
\ANN can be reduced to a small number of near neighbor queries
\cite{im-anntr-98, him-anntr-12}. Next, using locality sensitive
hashing they provide a data-structure that answers \ANN queries in
time (roughly) $\aproxOrder{n^{1/(1+\eps)}}$ and preprocessing time
and space $\aproxOrder{n^{1+1/(1+\eps)}}$; here the
$\aproxOrder{\cdot}$ hides terms polynomial in $\log n$ and
$1/\eps$. This was improved to $\aproxOrder{n^{1/(1+\eps)^2}}$ query
time, and preprocessing time and space
$\aproxOrder{n^{1+1/(1+\eps)^2}}$ \cite{ai-nohaa-06,
   ai-nohaa-08}. These bounds are near optimal
\cite{mnp-lblsh-06,owz-olblsh-11}.

In low dimensions (i.e., $\Re^d$), one can use linear space
(independent of $\eps$) and get \ANN query time $O(\log n +
1/\eps^{d-1})$ \cite{amnsw-oaann-98, h-gaa-11}. Interestingly, for
this data-structure, the approximation parameter $\eps$ is not
prespecified during the construction; one needs to provide it only
during the query.  An alternative approach is to use Approximate
Voronoi Diagrams (\AVD), introduced by Har-Peled \cite{h-rvdnl-01},
which are partition of space into regions, desirably of low
complexity, typically with a representative point for each region that
is an \ANN for any point in the region. In particular, Har-Peled
showed that there is such a decomposition of size $O\pth{(n
   /\eps^d)\log^2 n}$, such that \ANN queries can be answered in
$O(\log (n/\eps))$ time.  Arya and Malamatos \cite{am-lsavd-02} showed
how to build \AVD{}s of linear complexity (i.e., $O(n/\eps^d)$). Their
construction uses Well Separated Pair Decompositions
\cite{ck-dmpsa-95}. Further tradeoffs between query and space for
\AVD{}s were studied by Arya \etal \cite{amm-sttan-09}.

\myparagraph{Metric spaces}
One possible approach for the more general case when the data lies in
some abstract metric space, is to define a notion of dimension and
develop efficient algorithms in these settings. This approach is
motivated by the belief that real world data is ``low dimensional'' in
many cases, and should be easier to handle than true high dimensional
data. An example of this approach is the notion of \emphi{doubling
   dimension} \cite{a-pldr-83, h-lams-01,gkl-bgfld-03}.  The
\emphi{doubling constant} of metric space $\mtrA$ is the maximum, over
all balls $\ballC$ in the metric space $\mtrA$, of the minimum number
of balls needed to cover $\ballC$, using balls with half the radius of
$\ballC$. The logarithm of the doubling constant is the
\emphi{doubling dimension} of the space.  The doubling dimension can
be thought of as a generalization of the Euclidean dimension, as
$\Re^d$ has doubling dimension $\Theta(d)$.  Furthermore, the doubling
dimension extends the notion of growth restricted metrics of Karger
and Ruhl~\cite{kr-fnngr-02}.

The problem of \ANN in spaces of low doubling dimension was studied
before, see \cite{kr-fnngr-02, hkmr-nnng-04}. Talwar \cite{t-beald-04}
presented several algorithms for spaces of low doubling
dimension. Some of them were however dependent on the spread of the
point set.  Krauthgamer and Lee \cite{kl-nnsap-04} presented a net
navigation algorithm for \ANN in spaces of low doubling dimension.
Har-Peled and Mendel \cite{hm-fcnld-06} provided data-structures for
\ANN search that use linear space and match the bounds known for
$\Re^d$ \cite{amnsw-oaann-98}. Clarkson \cite{c-nnsms-06} presents
several algorithms for nearest neighbor search in low dimensional
spaces for various notions of dimensions.

\myparagraph{\ANN in high and low dimensions}

As indicated above, the \ANN problem is easy in low dimensions (either
Euclidean or bounded doubling dimension). If the dimension is high the
problem is considerably more challenging. There is considerable work
on \ANN in high dimensional Euclidean space (see \cite{im-anntr-98,
   kor-esann-00, him-anntr-12}) but the query time is only slightly
sublinear if $\eps$ is close to $0$. In general metric spaces, it is
easy to argue that (in the worst case) the \ANN algorithm must compute
the distance of the query point to all the input points.

It is natural to ask therefore what happens when the data (or the
queries) come from a low dimensional subspace that lies inside a high
dimensional ambient space. Such cases are interesting as it is widely
believed that in practice real world data usually lies on a low
dimensional manifold (or is close to lying on such a manifold). Such
low-dimensionality arises from the way the data is being acquired,
inherent dependency between parameters, aggregation of data that leads
to concentration of mass phenomena, etc.

Indyk and Naor \cite{in-nnpe-07} showed that if the data is in high
dimensional Euclidean space, but lies on a manifold with low doubling
dimension, then one can do a dimension reduction into constant
dimension (i.e., similar in spirit to the \JL lemma
\cite{jl-elmih-84}), such that $(1+\eps)$-\ANN to a query point (the
query point might lie anywhere in the ambient space) is preserved with
constant probability. Using an appropriate data-structure on the
embedded space and repeating this process sufficient number of times
results in a data-structure that can answer such \ANN queries in
polylog time (ignoring the dependency on $\eps$).

\myparagraph{The problem}
In this paper, we study the ``reverse'' problem. Here we are given a
high dimensional data set $\PntSet$, and we would like to preprocess
it for \ANN queries, where the queries come from a low-dimensional
subspace/manifold $\manifold$. The question arises naturally when the
given data is formed by merging together a large number of data sets,
while the \ANN queries come from a single data set.

In particular, the conceptual question here is whether this problem is low
or high dimensional in nature. Note that direct dimension reduction as
done by Indyk and Naor would not work in this case. Indeed, imagine
the data lies densely on a slightly deformed sphere in high
dimensions, and the query is the center of the sphere. Clearly, a
random dimension reduction via projection into constant dimension
would not preserve the $(1+\eps)$-\ANN.

\myparagraph{Our results}
Given a point set $\PntSet$ lying in a general metric space $\mtrA$
(which is not necessarily Euclidean and is conceptually high
dimensional), and a subspace $\manifold$ having low doubling
dimension $\dd$, we show how to preprocess $\PntSet$ such that given any
query point in $\manifold$ we can quickly answer $(1+\eps)$-\ANN
queries on $\PntSet$.  In particular, we get data-structures of
(roughly) linear size that answer $(1+\eps)$-\ANN queries in (roughly)
logarithmic time.

Our construction uses ideas developed for handling the low dimensional
case.  Initially, we embed $\PntSet$ and $\manifold$ into a space with
low doubling dimension that (roughly) preserves distances between
$\manifold$ and $\PntSet$. We can use the embedded space to answer
constant factor \ANN queries. Getting a better approximation requires
some further ideas.  In particular, we build a data-structure over
$\manifold$ that is somewhat similar to approximate Voronoi diagrams
\cite{h-rvdnl-01}. By sprinkling points carefully on the subspace
$\manifold$ and using the net-tree data-structure \cite{hm-fcnld-06}
we can answer $(1+\eps)$-\ANN queries in time $O( \eps^{-O(\dd)} +
2^{O(\dd)}\log n)$.

To get a better query time requires some further work.  In particular,
we borrow ideas from the simplified construction of Arya and Malamatos
\cite{am-lsavd-02} (see also \cite{amm-sttan-09}). Naively, this
requires us to use well separated pairs decomposition (i.e., \WSPD)
\cite{ck-dmpsa-95} for $\PntSet$. Unfortunately, no such small \WSPD
exists for data in high dimensions. To overcome this problem, we build
the \WSPD in the embedded space. Next, we use this to guide us in the
construction of the \ANN data-structure.  This results in a
data-structure that can answer $(1+\eps)$-\ANN queries in $O(
2^{O(\dd)} \log n)$ time.  See \secref{a:n:n} for details.

% SARIEL: new stuff.

We also present an algorithm for a weaker model, where the query
subspace is not given to us directly. Instead, every time an \ANN
query is issued, the algorithm computes a region around the query
point such that the returned point is a valid \ANN for all the points
in this region. Furthermore, the algorithm caches such regions, and
whenever a query arrives it first checks if the query point is already
contained in one of the regions computed, and if so it answers the
\ANN query immediately. Significantly, for this algorithm we need no
prespecified knowledge about the query subspace. The resulting
algorithm computes on the fly \AVD on the query subspace. In
particular, we show that if the queries come from a subspace with
doubling dimension $\dd$, then the algorithm would create at most $n/
\eps^{\porder{\dd}}$ regions overall.  A limitation of this new
algorithm is that we do not currently know how to efficiently perform
a point-location query in a set of such regions, without assuming
further knowledge about the subspace.  Interestingly, the new
algorithm can be interpreted as learning the underlying
subspace/manifold the queries come from.  See \secref{oann} for the
precise result.

\myparagraph{Organization}

In \secref{prelims}, we define some basic concepts, and as a warm-up
exercise study the problem where the subspace $\manifold$ is a linear
subspace of $\Re^d$ -- this provides us with some intuition for the
general case. We also present the embedding of $\PntSet$ and
$\manifold$ into the subspace $\mtrB$, which has low doubling
dimension while (roughly) preserving distances of interest.  In
\secref{constant}, we provide a data-structure for constant factor
\ANN using this embedding.  In \secref{epsilon:a:n:n:slow}, we use the
constant \ANN to get a data-structure for answering
$(1+\eps)$-\ANN. In \secref{a:n:n}, we use \WSPD to build a
data-structure that is similar in spirit to \AVD{}s. This results in a
data-structure with slightly faster \ANN query time.  The on the fly
construction of \AVD to answer \ANN queries without assuming any
knowledge of the query subspace is described in \secref{oann}.
Finally, conclusions are provided in \secref{conclusions}.

\section{Preliminaries}
\seclab{prelims}

\subsection{Problem and Model}
\subseclab{problem:model}

\myparagraph{The Problem}

We look at the \ANN problem in the following setting. Given a set
$\PntSet$ of $n$ data points in a metric space $\mtrA$, and a set
$\manifold \subseteq \mtrA$ of (hopefully low) doubling dimension
$\dd$, and $\eps > 0$, we want to preprocess the points of $\PntSet$,
such that given a query point $\query \in \manifold$ one can
efficiently find a $(1 + \eps)$-\ANN of $\query$ in $\PntSet$.

\myparagraph{Model}

We are given a metric space $\mtrA$ and a subset $\manifold \subseteq
\mtrA$ of doubling dimension $\dd$. We assume that the distance
between any pair of points can be computed in constant time in a
black-box fashion. Specifically, for any $\pntA, \pntQ \in \mtrA$ we
denote by $\dist{\pntA}{\pntQ}$ the distance between $\pntA$ and
$\pntQ$. We also assume that one can build nets on
$\manifold$. Specifically, given a point $\pntA \in \manifold$ and a
radius $r > 0$, we assume we can compute $2^{\dd}$ points $\pntA_i \in
\manifold$, such that $\ball{\pntA}{r} \cap \manifold \subseteq
\bigcup \ball{\pntA_i}{r/2}$. By applying this recursively we can
compute an \emphi{$r$-net} $\Net$ for any $\ball{\pntA}{R}$ centered
at $\pntA$; that is, for any point $\pntB \in \ball{\pntA}{R}$ there
exists a point $\pntC \in \Net$ such that $\dist{\pntB}{\pntC} \leq
r$. Let \algNet$(\pntA, R, r)$ denote this algorithm for computing
this $r$-net. The size of $\Net$ is $(R/r)^{O(\dd)}$, and we assume
this also bounds the time it takes to compute it. For example, in
Euclidean space $\Re^d$, let $\pntA$ be the origin and consider the
tiling of space by a grid of cubes of diameter $r$. One can compute an
$r$-net, by simply enumerating all the vertices of the grid cells that
intersect the cube $\pbrc[]{-R,R}^d$ surrounding $\ball{\pntA}{R} =
\ball{0}{R}$.

Finally, given any point $\pntA \in \mtrA$ we assume that one can
compute, in $\order{1}$ time, a point $\closest{\pntA} \in \manifold$
such that $\closest{\pntA}$ is the closest point in $\manifold$ to
$\pntA$. (Alternatively, $\closest{\pntA}$ might be specified for each
point of $\PntSet$ in advance.)

\myparagraph{Spread of a point set} 
For a point set $\PntSet$, the \emphi{spread} is the ratio $\frac{\max
   \limits_{\pntA,\pntB \in \PntSet}\dist[]{\pntA}{\pntB}} {\min
   \limits_{\pntA, \pntB \in \PntSet, \pntA \neq
      \pntB}\dist[]{\pntA}{\pntB}}$. The following result is
elementary.
\begin{lemma}
    Let $\manifold$ be a metric space of doubling dimension $\dd$ and
    $\PntSet \subseteq \manifold$ be a point set with spread
    $\lambda$. Then $\cardin{\PntSet} \leq \lambda^{\order[]{\dd}}$.
\end{lemma}

\myparagraph{Well separated pairs decomposition}

For a point set $\PntSet$, a \emphi{pair decomposition} of $\PntSet$
is a set of pairs $\WSPDRep = \brc{\MakeBig
   \brc{A_1,B_1},\ldots,\brc{A_s,B_s}}$, such that
\begin{inparaenum}[(I)]
    \item $A_i,B_i\subset \PntSet$ for every $i$,
    \item $A_i \cap B_i = \emptyset$ for every $i$, and
    \item $\cup_{i=1}^s A_i \otimes B_i = \PntSet \otimes \PntSet$.
\end{inparaenum}
Here $X \otimes Y = \brc{\MakeBig \brc{x,y} \sep x \in X, y \in Y,
   \text{ and } x \neq y}$.

A pair $\PntSetB \subseteq \PntSet$ and $\PntSetC \subseteq \PntSet$
is \emphic{$(1/\eps)$-separated}{separated!sets} if $\max \pth{
   \diameterX{\PntSetB}, \diameterX{\PntSetC} } \leq \eps \cdot
\dist{\PntSetB}{\PntSetC}$, where $\dist{\PntSetB}{\PntSetC} =
\min_{\pntA \in \PntSetB, \pntB \in \PntSetC} \dist{\pntA}{\pntB}$.
For a point set $\PntSet$, a \emphi{well-separated pair decomposition}
(\emphi{\WSPD{}}) of $\PntSet$ with parameter $1/\eps$ is a pair
decomposition of $\PntSet$ with a set of pairs $\WSPDRep = \brc{
   \brc{A_1,B_1},\ldots,\brc{A_s,B_s}}$, such that, for any $i$, the
sets $A_i$ and $B_i$ are $\eps^{-1}$-separated \cite{ck-dmpsa-95}.

\subsubsection{Net-trees}%
\seclab{net:tree}%

The net-tree \cite{hm-fcnld-06} is a data-structure that defines
hierarchical nets in finite metric spaces. Formally, a net-tree is
defined as follows: Let $\PntSet \subseteq \manifold$ be a finite
subset. A net-tree of $\PntSet$ is a tree $\nettree$ whose set of
leaves is $\PntSet$. Denote by $\PntSet_\vertexA$ the set of leaves in
the subtree rooted at a vertex $\vertexA \in \nettree$. With each
vertex $\vertexA$ is associated a point $\nettreerep{\vertexA} \in
\PntSet_\vertexA$. Internal vertices have at least two children. Each
vertex $\vertexA$ has a level $\lvl{\vertexA} \in \Integers \cup
\brc{-\infty}$. The levels satisfy $\lvl{\vertexA} <
\lvl{\parent{\vertexA}}$, where $\parent{\vertexA}$ is the parent of
$\vertexA$ in $\nettree$. The levels of the leaves are $-\infty$.  Let
$\nettreeconst$ be some large constant, say $\nettreeconst = 11$. The
following properties are satisfied:
\begin{inparaenum}[(I)]
    \item For every vertex $\vertexA \in \nettree$,
    $\ball{\nettreerep{\vertexA}} {\frac{2\nettreeconst}{\nettreeconst
          -1}\nettreeconst^{\lvl[]{\vertexA}}} \supseteq
    \PntSet_\vertexA$,
    \item For every vertex $\vertexA \in \nettree$ that is not the
    root, $\ball{\nettreerep{\vertexA}} {\frac{\nettreeconst -
          5}{2(\nettreeconst - 1)}
       \nettreeconst^{\lvl[]{\parent[]{\vertexA}} - 1}} \cap \PntSet
    \subseteq \PntSet_\vertexA$,
    \item For every internal vertex $\vertexB \in \nettree$, there
    exists a child $\vertexA \in \nettree$ of $\vertexB$ such that
    $\nettreerep{\vertexB} = \nettreerep{\vertexA}$.
\end{inparaenum}

\subsection{Warm-up exercise: Affine Subspace}

% [width=0.6\linewidth,clip=]
\begin{figure*}[t]
    \centerline{
       \includegraphics{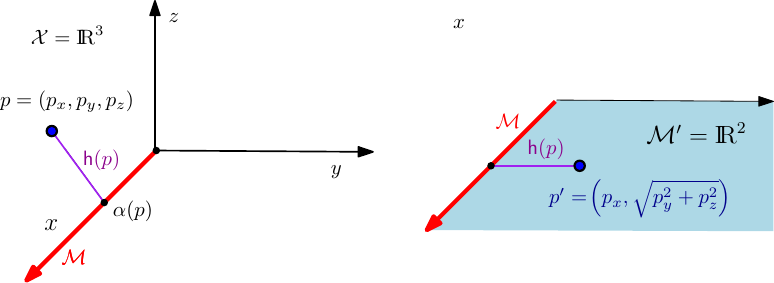}
    }
    \caption{An example of embedding of space into two dimensions
       where $\manifold$ is the $x$-axis.}
    \figlab{fig:euc}
\end{figure*}

We first consider the case where our query subspace is an affine
subspace embedded in $d$ dimensional Euclidean space. Thus let $\mtrA
= \Re^d$ with the usual Euclidean metric. Suppose our query subspace
$\manifold$ is an affine subspace of dimension $k$ where $k \ll d$.
We are also given $n$ data points $\PntSet = \brc{\pntA_1, \pntA_2,
   \ldots, \pntA_n}$. We want to preprocess $\PntSet$ such that given
a $\query \in \manifold$ we can quickly find a point $\pntA_i \in
\PntSet$ which is a $(1+\eps)$-\ANN of $\query$ in $\PntSet$.

We choose an orthonormal system of coordinates for $\manifold$. Denote
the projection of a point $\pntA$ to $\manifold$ as
$\closest{\pntA}$. Denote the coordinates of a point $\closest{\pntA}
\in \manifold$ in the chosen coordinate system as $(\pntA^1,
\pntA^2,\dots,\pntA^k)$. Let $\height{\pntA}$ denote the distance of
any $\pntA \in \Re^d$ from the subspace $\manifold$. Notice that
$\height{\pntA} = \normX{\pntA}{ \closest{\pntA}}$, and consider the
following embedding.

\begin{definition}
    For the point $\pntA \in \Re^d$, the embedded point is
    $\mapped{\pntA} = \pth[]{\MakeSBig
       \pntA^1,\pntA^2,\dots,\pntA^k,\height{\pntA}} \in \Re^{k+1}$.
\end{definition}

An example of the above embedding is shown in \figref{fig:euc}.  It is
easy to see that for $x \in \manifold$ and $y \in \Re^d$, by the
Pythagorean theorem, we have $\normX{x}{y}^2 =
\normX{x}{\closest{y}}^2 + \normX{\closest{y}}{y}^2 =
\normX{x}{\closest{y}}^2 + \height{y}^2 =
\normX{\mapped{x}}{\mapped{y}}^2$.  So,
$\normX{x}{y}=\normX{\mapped{x}}{\mapped{y}}$. That is, the above
embedding preserves the distances between points on $\manifold$ and
any point in $\Re^d$.

As such, given a query point $\query \in \manifold$, let
$\mapped{\pntA_i}$ be its $(1 + \eps)$-\ANN in $\mapped{\PntSet}
\subseteq \Re^{k+1}$. Then the original point $\pntA_i \in \PntSet$
(that generated $\mapped{\pntA_i}$) is a $(1 + \eps)$-\ANN of $\query$
in the original space $\Re^d$.

But this is easy to do using known data-structures for \ANN
\cite{amnsw-oaann-98}, or the data-structures for approximate Voronoi
diagram \cite{h-rvdnl-01, am-lsavd-02}.

Thus, we have $n$ points in $\Re^{k+1}$ to preprocess and, without
loss of generality, we can assume that $\mapped{\pntA_i}$ are all
distinct. Now given $\eps \leq 1/2$, we can preprocess the points
$\brc{\mapped{\pntA_1},\dots,\mapped{\pntA_n}}$ and construct an
approximate Voronoi diagram consisting of $\order{n\eps^{-(k+1)} \log
   \eps^{-1}}$ regions \cite{am-lsavd-02}.  Each such region is the
difference of two cubes. Given a point $\mapped{\query} \in \Re^{k+1}$
we can find a $(1 + \eps)$-\ANN in $\order{\log \pth{ \nfrac{n}{\eps}
   } }$ time, using this data-structure.

\subsection{An Embedding}
\seclab{embedding}
\begin{figure*}[t]
    \centerline{
       \includegraphics{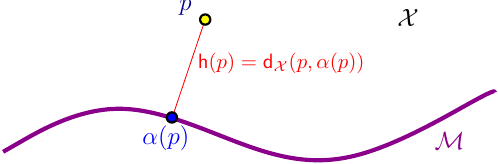}
    }
    \caption{The quantities $\closest{\pntA}$ and $\height{\pntA}$.}
    \figlab{fig:manifold}
\end{figure*}

Here, we show how to embed the points of $\PntSet$ (and all of
$\mtrA$) into another metric space $\mtrB$ with finite doubling
dimension, such that the distances between $\PntSet$ and $\manifold$
are roughly preserved.

For a point $\pntA \in \mtrA$, let $\closest{\pntA}$ denote the
closest point in $\manifold$ to $\pntA$ (for the sake of simplicity of
exposition we assume this point is unique).  The \emphi{height} of a
point $\pntA \in \mtrA$ is the distance between $\pntA$ and
$\closest{\pntA}$; namely, $\height{\pntA} =
\distM{\mtrA}{\pntA}{\closest{\pntA}}$. For a set $\setA \subseteq
\mtrA$, let $\closest{\setA}$ denote the set $\brc{\closest{x} \sep{ x
      \in \setA}}$.  An example is shown in \figref{fig:manifold}.

\begin{definition}[\emphi{$\mapped{\manifold}$ embedding}.]
    Consider the embedding of $\mtrA$ into $\mtrB = \manifold \times
    \Re^+$ induced by the distances of points of $\mtrA$ from
    $\manifold$. Formally, for a point $\pntA \in \mtrA$, the
    embedding is defined as
    \begin{align*}
        \mapped{\pntA} = \pth{\MakeSBig \closest{\pntA}, \,
           \height{\pntA}} \in \mtrB .
    \end{align*}
    % The resulting metric space $\mtrB$ is $\manifold \times
    % \Re^{+}$.
\end{definition}

The distance between any two points $\mapped{\pntA} =
\pth{\closest{\pntA}, \height{\pntA}}$ and $\mapped{\pntB} =
\pth{\closest{\pntB}, \height{\pntB}}$ of $\mtrB$ is defined as
\begin{align*}
    \distM{\mtrB}{\mapped{\pntA}}{\mapped{\pntB}}%
    =%
    \distM{\mtrA}{\MakeSBig \closest{\pntA}}{\closest{\pntB}}
    + \cardin{ \height{\pntA} - \height{\pntB} }.
\end{align*}
It is easy to verify that $\distM{\mtrB}{\cdot}{\cdot}$ complies with
the triangle inequality. For the sake of simplicity of exposition, we
assume that for any two distinct points $\pntA$ and $\pntB$ in our
(finite) input point set $\PntSet$ it holds that $\mapped{\pntA} \ne
\mapped{\pntB}$ (that is,
$\distM{\mtrB}{\mapped{\pntA}}{\mapped{\pntB}} \ne 0$). This can be
easily guaranteed by introducing symbolic perturbations.

\begin{lemma}%
    \lemlab{embed:low:d:d}%
    The following holds:
    \begin{inparaenum}[(A)]
        \item For any two points $x,y \in \manifold$, we have
        $\distM{\mtrB}{\mapped{x}}{\mapped{y}} = \distM{\mtrA}{x}{y}$.

        \item%
        \lempntlab{embed:low:d:d:B}%
        For any point $x \in \manifold$ and $y \in \mtrA$, we have
        $\distM{\mtrA}{x}{y} \leq
        \distM{\mtrB}{\mapped{x}}{\mapped{y}} \leq 3
        \distM{\mtrA}{x}{y}$.

        \item The space $\mtrB$ has doubling dimension at most $2\dd +
        2$, where $\dd$ is the doubling dimension of $\manifold$.
        % \
    \end{inparaenum}
\end{lemma}%
\begin{proof}
    (A) Clearly, for $x, y \in \manifold$, we have $\mapped{x} =
    (x,0)$ and $\mapped{y} = (y,0)$. As such,
    $\distM{\mtrB}{\mapped{x}}{\mapped{y}} = \distM{\mtrA}{x}{y} +
    |0-0| = \distM{\mtrA}{x}{y}$.
    
    (B) Let $x \in \manifold$ and $y \in \mtrA$. We have $\mapped{x} =
    (x,0)$ and $\mapped{y} = \pth{\MakeSBig
       \closest{y},\distM{\mtrA}{y}{\closest{y}\MakesBig}}$. As such,
    \begin{align*}
        \distM{\mtrB}{\mapped{x}}{\mapped{y}} &=%
        \distM{\mtrA}{\closest{x}}{\closest{y}} + \cardin{0
           -\height{y}}%
        =%
        \distM{\mtrA}{x}{\closest{y}} + \distM{\mtrA}{\closest{y}}{y}%
        \geq%
        \distM{\mtrA}{x}{y},
    \end{align*}
    by the triangle inequality. On the other hand, because
    $\distM{\mtrA}{y}{\closest{y}} = \distM{\mtrA}{y}{\manifold} \leq
    \distM{\mtrA}{x}{y}$, we have
    \begin{align*}
        \distM{\mtrB}{\mapped{x}}{\mapped{y}}%
        &=%
        \distM{\mtrA}{\MakeSBig \closest{x}}{\closest{y}} + \cardin{
           \height{x} - \height{y} }%
        =%
        \distM{\mtrA}{\MakeSBig x }{\closest{y}} + \height{y} %
        \\%
        &=%
        \distM{\mtrA}{x}{\closest{y}} + \distM{\mtrA}{y}{\closest{y}}
        \leq%
        \pth{\MakeSBig \distM{\mtrA}{x}{y} +
           \distM{\mtrA}{y}{\closest{y}}} +
        \distM{\mtrA}{y}{\closest{y}}%
        \\
        &=%
        \distM{\mtrA}{x}{y} + 2\distM{\mtrA}{y}{\closest{y}}%
        \leq%
        3 \distM{\mtrA}{x}{y},
    \end{align*}
    by the triangle inequality.
    
    (C) Consider a point $(\pntA, \num) \in \manifold \times \Re^+ =
    \mtrB $ and the ball $\ballC = \ballExt{\mtrB}{(\pntA,\num)}{r}
    \subseteq \mtrB$ of radius $r$ centered at
    $(\pntA,\num)$. Consider the projection of $\ballC$ into
    $\manifold$; that is $\PntSet_{\manifold} = \brc{ \pntB \sep{
          \pth{\pntB, h} \in \ballC }}$. Similarly, let $\PntSet_{\Re}
    = \brc{ h \sep{ \pth{\pntB, h} \in \ballC }}$.
    
    Clearly, $\ballExt{\mtrB}{\MakesBig (\pntA, \num)}{r} \subseteq
    \PntSet_{\manifold} \times \PntSet_{\Re}$, and
    $\PntSet_{\manifold}$ is contained in
    $\ballExt{\manifold}{\pntA}{r} = \ballExt{\mtrA}{\pntA}{r} \cap
    \manifold$. Since the doubling dimension of $\manifold$ is $\dd$,
    this ball can be covered by $2^{2 \dd}$ balls of the form
    $\ballExt{\manifold}{\pntA_i}{r/4}$ with centers $\pntA_i \in
    \manifold$.
    
    Also since $\PntSet_{\Re} \subseteq \Re$ is contained in the
    interval $\pbrc[]{\num -r, \num + r}$ having length $2r$, it can
    be covered by at most $4$ intervals $I_1, \ldots, I_4$ of length
    $r/2$ each, centered at values ${x_1, \ldots, x_4}$,
    respectively. (Intuitively, each of the intervals $I_j$, is a
    ``ball'' of radius $r/4$.) Then,
    \begin{align*}
        \ballExt{\mtrB}{ \MakeBig (\pntA,\num)}{r} %
        &\quad\subseteq\quad%
        \PntSet_{\manifold} \times \PntSet_{\Re} %
        \quad\subseteq\quad%
        \pth{\bigcup_i \ballExt{\manifold}{\pntA_i}{r/4} } \times
        \pth{ \bigcup_{j=1}^4 I_j }\\
        % \\%
        %
        &\quad\subseteq\quad%
        \bigcup_{j=1}^4 \bigcup_{i}
        \pth{\ballExt{\manifold}{\pntA_i}{r/4} \times I_j} %\\%
        \quad\subseteq\quad%
        \bigcup_{j=1}^4 \bigcup_{i} \ballExt{\mtrB}{\MakeSBig
           \pth{\pntA_i,x_j}}{r/2},
    \end{align*}
    since the set $\ballExt{\manifold}{\MakesBig\pntA_i}{r/4} \times
    I_j$ is contained in
    $\ballExt{\mtrB}{\MakesBig(\pntA_i,x_j)}{r/2}$. We conclude that
    $\ballExt{\mtrB}{(\pntA, \num)}{r}$ can be covered using at most
    $2^{2\dd + 2}$ balls of half the radius.
\end{proof}

% ----------------------------------------------------------------------
% ----------------------------------------------------------------------
% ----------------------------------------------------------------------

\section{A Constant Factor \ANN Algorithm}
\seclab{constant} In this section we present a $6$-\ANN algorithm. We
refine this to a $(1+\eps)$-\ANN in the next section.

\myparagraph{Preprocessing}
In the preprocessing stage, we map the points of $\PntSet$ into the
metric space $\mtrB$ of \lemref{embed:low:d:d}. Build a net-tree for the
point set $\mapped{\PntSet} = \brc{\mapped{\pntA} \sep{\pntA \in
      \PntSet}}$ in $\mtrB$ and preprocess it for \ANN queries using
the net-tree data-structure (augmented for nearest neighbor queries) 
of Har-Peled and Mendel \cite{hm-fcnld-06}. Let $\DS$ denote the 
resulting data-structure.

\myparagraph{Answering a query}
Given $\query \in \manifold$, we compute a $2$-\ANN to
$\mapped{\query} \in \mtrB$ using $\DS$. Let this be the point
$\mapped{\algonnA}$.  Return $\distM{\mtrA}{\query}{\algonnA}$, where
$\algonnA$ is the original point in $\PntSet$ corresponding to
$\mapped{\algonnA}$.

\myparagraph{Correctness}
Let $\nnA$ be the nearest neighbor of $\query$ in $\PntSet$ and let
$\algonnA$ be the point returned. As $\query \in \manifold$ we have by
\lemref{embed:low:d:d} \lempntref{embed:low:d:d:B} that
$\distM{\mtrA}{\query}{\algonnA} \leq
\distM{\mtrB}{\mapped{\query}}{\mapped{\algonnA}}$ and
$\distM{\mtrB}{\mapped{\query}}{\mapped{\nnA}} \leq 3
\distM{\mtrA}{\query}{\nnA}$. As $\mapped{\algonnA}$ is a $2$-\ANN for
$\mapped{\query}$ it follows,
\begin{equation*}
    \distM{\mtrA}{\query}{\algonnA} \leq 
    \distM{\mtrB}{\mapped{\query}}{\mapped{\algonnA}} \leq 
    2 \distM{\mtrB}{\mapped{\query}}{\mapped{\nnA}}
    \leq  6 \distM{\mtrA}{\query}{\nnA}.
\end{equation*}
We thus proved the following.

\begin{lemma}%
    \lemlab{constant:ANN}%
    Given a set $\PntSet \subseteq \mtrA$ of $n$ points and a subspace
    $\manifold$ of doubling dimension $\dd$, one can build a
    data-structure in $2^{\porder{\dd}}n \log n$ expected time, such
    that given a query point $\query \in \manifold$, one can return a
    $6$-\ANN to $\query$ in $\PntSet$ in $2^{\porder{\dd}} \log n$
    query time.  The space used by this data-structure is
    $2^{\porder{\dd}}n$.
\end{lemma}

\begin{proof}
    Since the doubling dimension of $\mtrB$ is at most $2 \dd + 2$,
    building the net-tree and preprocessing it for \ANN queries takes
    $2^{\porder{\dd}} n \log n$ expected time, and the space used is
    $2^{\porder{\dd}}n$ \cite{hm-fcnld-06}. The $2$-\ANN query for a
    point $\query$ takes time $2^{\porder{\dd}} \log n$.
\end{proof}

% ----------------------------------------------------------------------
% ----------------------------------------------------------------------
% ----------------------------------------------------------------------

\section{Answering $(1+\eps)$-\ANN}
\seclab{epsilon:a:n:n:slow}

Once we have a constant factor approximation to the nearest-neighbor
in $\PntSet$ it is not too hard to boost it into $(1+\eps)$-\ANN. To
this end we need to understand what the net-tree \cite{hm-fcnld-06}
provides us with. See Har-Peled and Mendel \cite{hm-fcnld-06} (see
also \secref{net:tree}) for a precise definition of the net-tree.
Roughly speaking, the nodes at a given level $l$, define an
$\nettreeconst^l$-net for $\PntSetB$.  This means that one can compute
an $r$-net for any desired $r$ by looking at nodes whose levels define
the right resolution. Thus $r$-nets derived from the net-tree have a
corresponding set of nodes in the net-tree.  Suppose one needs to
find an $r$-net for the points of $\PntSetB$ inside a ball
$\ballExt{\manifold}{\pntA}{R}$. One computes an \ANN $\algonnA \in
\PntSetB$ of the center $\pntA$. This determines a leaf node $l$ of
the net-tree.  One then seeks out a vertex $\vertexA$ of the net-tree
on the $l$ to root path, such that $l \in \PntSetB_\vertexA$ and the
$\vertexA$ associated ball radius is roughly $R$. By adding
appropriate pointers, one can perform this hopping up the tree in
logarithmic time.  Now, exploring the top of the subtree rooted at
$\vertexA$, and collecting the representative points of the vertices
in that traversal, one can compute an $r$-net for the points in
$\PntSetB \cap \ballExt{\manifold}{\pntA}{R}$. In particular, using
the \ANN data-structure of Har-Peled and Mendel \cite{hm-fcnld-06}
this operation is readily supported.
\begin{lemma}[\cite{hm-fcnld-06}]%
    \lemlab{net:quick}%
    Given a net-tree for a set $\PntSetB \subseteq \manifold$ of $n$
    points in a metric space with doubling dimension $\dd$, and given
    a point $\pntA \in \manifold$ and radius $r \leq R$, one can
    compute an $r$-net $\Net \subseteq \PntSetB$ of $\PntSetB \cap
    \ballExt{\manifold}{\pntA}{R}$, such that the following properties
    hold:
    \begin{compactenum}[(A)]
        \item For any point $\pntB \in \PntSetB \cap
        \ballExt{\manifold}{\pntA}{R}$ there exists a point $\pntC \in
        \Net$ such that $\distM{\manifold}{\pntB}{\pntC} \leq r$.
        
        \item $\cardin{\Net} = (R/r)^{O(\dd)}$.
        
        \item Each point $\pntD \in \Net$ corresponds to a node
        $v(\pntD)$ in the net-tree. Let $\PntSetB_{v(\pntD)}$ denote
        the subset of points of $\PntSetB$ stored in the subtree of
        $v(\pntD)$.  The union $\bigcup_{\pntD \in \Net}
        \PntSetB_{v(\pntD)}$ covers $\PntSetB \cap
        \ballExt{\manifold}{\pntA}{R}$.
        
        \item For any $\pntD \in \Net$, the diameter of the point set
        $\PntSetB_{v(\pntD)}$ is bounded by $r$.
        
        \item The time to compute $\Net$ is $2^{O(\dd)} \log n +
        O\pth{\cardin{\Net} }$.
    \end{compactenum}
\end{lemma}

\myparagraph{Construction}
For every point $\pntA \in \PntSet$ we compute an $r(\pntA)$-net
$U(\pntA)$ for $\ballExt{\manifold}{\closest{\pntA}}{\brad{\pntA}}$,
where $r(\pntA) = \eps \height{\pntA} /(20\constA)$ and $\brad{\pntA}
=\constA\height{\pntA}/\eps$. Here $\constA$ is some sufficiently
large constant. This net is computed using the algorithm \algNet, see
\subsecref{problem:model}.  This takes $1/\eps^{O(\dd)}$ time to
compute for each point of $\PntSet$.

For each point $\pntC$ of the net $U(\pntA) \subseteq \manifold$ store
the original point $\pntA$ it arises from, and the distance to the
original point $\pntA$. We will refer to $\reachX{\pntC} =
\distM{\mtrA}{\pntC}{\pntA}$ as the \emphi{reach} of $\pntC$.

Let $\PntSetB \subseteq \manifold$ be union of all these
nets. Clearly, we have that $\cardin{\PntSetB} =
n/\eps^{O(\dd)}$. Build a net-tree $\Tree$ for the points of
$\PntSetB$.  We compute in a bottom-up fashion for each node $v$ of
the net-tree $\Tree$ the point with the smallest reach stored in
$\PntSetB_v$.

\myparagraph{Answering a query}
Given a query point $\query \in \manifold$, compute using the
algorithm of \lemref{constant:ANN} a $6$-\ANN to $\query$ in
$\PntSet$. Let $\Delta$ be the distance from $\query$ to this
\ANN. Let $R = 20 \Delta$, and $r' = \eps \Delta /20$. Using $\Tree$
and \lemref{net:quick}, compute an $r'$-net $N$ of
$\ballExt{\manifold}{\query}{R} \cap \PntSetB$.

Next, for each point $\pntA \in \Net$ consider its corresponding
node $v(\pntA) \in \Tree$. Each such node stores a point of minimum
reach in $\PntSetB_{v(\pntA)}$. We compute the distance to each such
minimum-reach point and return the nearest-neighbor found as the \ANN.

\begin{theorem}
    Given a set $\PntSet \subseteq \mtrA$ of $n$ points and a subspace
    $\manifold$ of doubling dimension $\dd$, and a parameter $\eps>0$,
    one can build a data-structure in $n \eps^{-\porder{\dd}} \log n$
    expected time, such that given a query point $\query \in
    \manifold$, one can return a $(1+\eps)$-\ANN to $\query$ in
    $\PntSet$. The query time is $2^{\porder{\dd}} \log n +
    \eps^{-\porder{\dd}}$.  This data-structure uses $n
    \eps^{-\porder{\dd}}$ space.
    
    \thmlab{1:epsilon:ANN}
\end{theorem}

\begin{proof}
    We only need to prove the bound on the quality of the
    approximation.  Consider the nearest-neighbor $\nnA$ to $\query$
    in $\PntSet$.
    \begin{compactenum}[(A)]
        \item If there is a point $\pntD \in U(\nnA) \subseteq
        \PntSetB$ within distance $r'$ from $\query$ then there is
        a net point $\pntC$ of $N$ that contains $\pntD$ in its
        subtree of $\Tree$. Let $\pntE_\algonnA$ be the point of
        minimum reach in $\PntSetB_{v(\pntC)}$, and let $\algonnA \in
        \PntSet$ be the corresponding original point.  Now, we have
        \begin{equation*}
            \distM{\mtrA}{\query}{\algonnA}  
            \leq 
            \distM{\mtrA}{\query}{\pntE_\algonnA}
            + \distM{\mtrA}{\pntE_\algonnA}{\algonnA}
            \leq \distM{\mtrA}{\query}{\pntE_\algonnA} + 
            \distM{\mtrA}{\pntD}{\nnA}
        \end{equation*}
        as the point $\pntE_\algonnA$ has reach
        $\distM{\mtrA}{\pntE_\algonnA}{\algonnA}$, $\pntE_\algonnA$ is
        the point of minimal reach among all the points of
        $\PntSetB_{v(\pntC)}$, $\pntD \in \PntSetB_{v(\pntC)}$, and
        $\distM{\mtrA}{\pntD}{\nnA}$ is the reach of $\pntD$ and thus
        an upper bound on
        $\distM{\mtrA}{\pntE_\algonnA}{\algonnA}$. By the triangle
        inequality, we have
        \begin{align*}
            \distM{\mtrA}{\query}{\algonnA}%
            &\leq%
            \distM{\mtrA}{\query}{\pntE_\algonnA}+\distM{\mtrA}{\query}{\nnA}
            + \distM{\mtrA}{\pntD}{\query}%
            \\
            &\leq%
            \pth{\MakeSBig\distM{\mtrA}{\query}{\pntD} +
               \distM{\mtrA}{\pntD}{\pntE_\algonnA}}
            +\distM{\mtrA}{\query}{\nnA} +
            \distM{\mtrA}{\pntD}{\query}%
            \\
            &\leq%
            \distM{\mtrA}{\query}{\nnA} + 3r',
        \end{align*}
        as $\pntD, \pntE_\algonnA \in \PntSetB_{v(\pntC)}$, the
        diameter of $\PntSetB_{v(\pntC)}$ is at most $r'$, and by
        assumption $\distM{\mtrA}{\pntD}{\query} \leq r'$. So we have,
        \begin{equation*}
            \distM{\mtrA}{\query}{\algonnA} %
            \leq 
            \distM{\mtrA}{\query}{\nnA} + 3 \eps \Delta/ 20
            \leq (1 + \eps) \distM{\mtrA}{\query}{\nnA}.
        \end{equation*}
        
        \item Otherwise, it must be that,
        $\distM{\mtrA}{\query}{U(\nnA)} > r'$. Observe that it must
        be that $r(\nnA) < r'$ as $\height{\nnA} \leq \Delta$. It must
        be therefore that the query point is outside the region
        covered by the net $U(\nnA)$. As such, we have
        \begin{align*}
            R(\nnA) 
            &=
            \frac{\constA \height{\nnA}}{\eps} < 
            \distM{\mtrA}{\closest{\nnA}}{\query}
            \leq
            \distM{\mtrA}{\query}{\nnA} + \distM{\mtrA}{\nnA}{\closest{\nnA}}
            \leq
            2 \distM{\mtrA}{\nnA}{\query} \leq 2\Delta,
        \end{align*}
        which means $\height{\nnA} \leq 2 \eps \Delta / \constA$.
        Namely, the height of the point $\nnA$ is insignificant in
        comparison to its distance from $\query$ (and conceptually can
        be considered to be zero). In particular, consider the net
        point $\pntC \in \Net$ that contains in its subtree
        the point $\pntD \in U(\nnA)$  
        closest to $\closest{\nnA}$ i.e. 
        $\distM{\manifold}{\closest{\nnA}}{\pntD} \leq r(\nnA)$. 
        The point of smallest reach in this subtree provides
        a $(1+\eps)$-\ANN as an easy but tedious argument similar to
        the one above shows.
    \end{compactenum}
    \aftermathA
\end{proof}
\section{Answering $(1+\eps)$-\ANN faster}
\seclab{a:n:n}

In this section, we extend the approach used in the above construction
to get a data-structure which is similar in spirit to an \AVD of
$\PntSet$ on $\manifold$. Specifically, we spread a set of points
$\centerset$ on $\manifold$, and we associate a point of $\PntSet$
with each one of them.  Now, answering $2$-\ANN on $\centerset$, and
returning the point of $\PntSet$ associated with this point, results
in the desired $(1+\eps)$-\ANN.

\subsection{The construction}

For a set $\mapped{Z} \subseteq \mapped{\PntSet}$ let
\begin{equation*}
    \hMax{\mapped{Z}} = \max_{(\pntA,h) \in \mapped{Z}} h.
\end{equation*}

The preprocessing stage is presented in \figref{preprocessing}, and
the algorithm for finding the $(1+\eps)$-\ANN for a given query is
presented in \figref{a:n:n}.

\begin{figure*}
    \centerline{
       \begin{algorithm}
           \algBuildANN{}$(\PntSet, \manifold)$. \+\+ \\
           $\mapped{\PntSet} = \brc{\mapped{x} \sep{ x \in
                 \PntSet}} \subseteq \mtrB$\\
           Compute a $8$-\WSPD $\WSPDRep =\brc{
              \brc{\mapped{A_1},\mapped{B_1}},\ldots,\brc{\mapped{A_s},
                 \mapped{B_s}}}$ of $\mapped{\PntSet}$\\
           \For
           $\brc{\mapped{A_i},\mapped{B_i}} \in \WSPDRep$ \Do \+ \\
           Choose points $\mapped{a_i} \in \mapped{A_i}$ and
           $\mapped{b_i} \in \mapped{B_i}$.\\
           $\pDist_i =
           \distM{\mtrB}{\mapped{a_i}}{\mapped{b_i}}$,\;\;\; $\PDist_i
           = \pDist_i + \hMax{\mapped{A_i}} +
           \hMax{\mapped{B_i}}$\\
           $R_i = \constB \PDist_i / \eps$, \;\; $r_i = \eps \PDist_i/\constB$ \\
           $\Net_i = \algNet( \closest{a_i}, R_i, r_i ) \cup
           \algNet( \closest{b_i}, R_i, r_i )$.\-\\
           \\
           $\centerset = \Net_1 \cup \ldots \cup \Net_s$\\
           $\centernettree \leftarrow$ Net-tree for $\centerset$
           \cite{hm-fcnld-06}\\
           \For $\pntA \in \centerset$ \Do \+\\
           Compute $\NNbrInSet{\pntA}{\PntSet}$ and store it with
           $\pntA$
       \end{algorithm}}
    \caption{Preprocessing the subspace $\manifold$ to answer
       $(1+\eps)$-\ANN queries on $\PntSet$. Here $\constB$ is a
       sufficiently large constant.}
    \figlab{preprocessing}
\end{figure*}

\begin{figure}
    \centerline{
       \begin{algorithm}
           \algANN( $\query \in \manifold$ )\+\\
           $\pntA \leftarrow$ $2$-\ANN of $\query$ in $\centerset$\\
           (Use net-tree $\centernettree$ \cite{hm-fcnld-06} to
           compute
           $\pntA$.)\\
           $\algonnA \leftarrow $ the point in $\PntSet$
           associated with $\pntA$.\\
           \Return $\algonnA$
       \end{algorithm}%
    }
    \caption{Computing a $(1 + \order{\eps})$-\ANN in $\PntSet$ for a
       query point $\query \in \manifold$.}
    \figlab{a:n:n}
\end{figure}

\subsection{Analysis}

Suppose the data-structure returned $\algonnA$ and the actual nearest
neighbor of $\query$ is $\nnA$. If $\algonnA = \nnA$ then the
algorithm returned the exact nearest-neighbor to $\query$ and we are
done. Otherwise, by our general position assumption, we can assume
that $\mapped{\algonnA} \neq \mapped{\nnA}$.  Note that there is a
\WSPD pair $\brc{\mapped{\sA},\mapped{\sB}}\in \WSPDRep$ that
separates $\mapped{\algonnA}$ from $\mapped{\nnA}$ in $\mtrB$; namely,
$\mapped{\algonnA} \in \mapped{\sA}$ and $\mapped{\nnA} \in
\mapped{\sB}$.  Let
\begin{align*}
    \pDist = \distM{\mtrB}{\mapped{\pa}}{\mapped{\pb}},
\end{align*}
where $\mapped{\pa}$ and $\mapped{\pb}$ are the representative points
of $\mapped{\sA}$ and $\mapped{\sB}$, respectively.  Let $\pa$ and
$\pb$ be the points of $\PntSet$ corresponding to $\mapped{\pa}$ and
$\mapped{\pb}$, respectively. Now, let
\begin{align*}
    \PDist = \hMax{\mapped{\sA}} + \hMax{\mapped{\sB}} + \pDist,%
    \qquad%
    R = \constB \PDist /\eps%
    \qquad%
    \text{and}%
    \qquad%
    r = \eps \PDist/ \constB.
\end{align*}
% No, let

\begin{observation}%
    \obslab{diam}%
    By the definition of a $8$-\WSPD and the triangle inequality, for
    any $\mapped{x} \in \mapped{\sA}$ and $\mapped{y} \in
    \mapped{\sB}$, we have that $\distM{\mtrB}{\mapped{x}}{\mapped{y}}
    \leq \diameterX{\mapped{\sA}} + \diameterX{\mapped{\sB}} +
    \distM{\mtrB}{\mapped{\pa}}{\mapped{\pb}} \leq \pth[]{5/4}
    \pDist$.
\end{observation}

We study the two possible cases, $\query \notin
\ballExt{\manifold}{\closest{{\pa}}}{ R }$ $\cup$
$\ballExt{\manifold}{\closest{{\pb}}}{ R }$ (\lemref{c:notin}) and
$\query \in \ballExt{\manifold}{\closest{{\pa}}}{R}$ $\cup$
$\ballExt{\manifold}{\closest{{\pb}}}{R}$ (\lemref{c:in}).
\begin{lemma}%
    \lemlab{c:notin}%
    If $\query \notin \ballExt{\manifold}{\closest{{\pa}}}{ R } \cup
    \ballExt{\manifold}{\closest{{\pb}}}{ R }$ then the algorithm from
    \figref{a:n:n} returns a $(1 + \eps)$-\ANN in $\PntSet$ to the
    query point $\query$ (assuming $\constB$ is sufficient
    large). Restated informally -- if $\query$ is far from both
    $\algonnA$ and $\nnA$ (compared to the distance between them) then
    the \ANN computed is correct.
\end{lemma}
\begin{proof}
    We have $\distM{\mtrA}{\closest{\nnA}}{\closest{\algonnA}} \leq
    \distM{\mtrB}{\mapped{\nnA}}{\mapped{\algonnA}}%
    \leq 5/4 \pDist$ by \obsref{diam}. So, by the triangle inequality,
    we have $\distM{\mtrA}{\nnA}{\algonnA}%
    \leq %
    \height{{\nnA}} + %
    \distM{\mtrA}{\closest{\nnA}}{\closest{\algonnA}} +
    \height{{\algonnA}} \leq \hMax{\mapped{\sA}} +(5/4)\pDist +
    \hMax{\mapped{\sB}}
    \leq (\nfrac{5}{4}) \PDist$.

    Since $\mapped{\nnA}, \mapped{\pb} \in \mapped{\sB}$, we have
    $\distM{\mtrA}{\closest{\nnA}}{\closest{\pb}}%
    \leq%
    \distM{\mtrB}{\mapped{\nnA}}{\mapped{\pb}}%
    \leq%
    \diameterX{\mapped{\sB}} \leq \pDist / 8 \leq \PDist / 8$.
    Therefore,
    \begin{align*}
        \distM{\mtrA}{\query}{\closest{\nnA}}%
        &\geq%
        \distM{\mtrA}{\query}{\closest{{\pb}}} -
        \distM{\mtrA}{\closest{\nnA}}{\closest{{\pb}}}%
        \geq%
        R - \diameterX{\mapped{\sB}}%
        = \constB \frac{\PDist}{\eps} - \diameterX{\mapped{\sB}}%
           \\%
           &%
        \geq%
        \PDist \pth{\frac{\constB}{\eps} - \frac{1}{8}}%
        \geq%
        \frac{\constB \PDist}{2\eps},
    \end{align*}
    assuming $\eps \leq 1$ and $\constB \geq 1$.  Now,
    $\distM{\mtrA}{\query}{\nnA} \geq \distM{\mtrA}{\nnA}{\manifold} =
    \distM{\mtrA}{\nnA}{\closest{\nnA}} $, and thus by the triangle
    inequality, we have
    \begin{align*}
        \distM{\mtrA}{\query}{\nnA}%
        \geq%
        \frac{\distM{\mtrA}{\query}{\nnA} +
           \distM{\mtrA}{\nnA}{\closest{\nnA}} }{2}%
        \geq%
        \frac{\distM{\mtrA}{\query}{\closest{\nnA}}}{2}%
        \geq%
        \frac{\constB \PDist}{4 \eps}.
    \end{align*}
    This implies that $\distM{\mtrA}{\query}{y}%
    \leq%
    \distM{\mtrA}{\query}{\nnA} + \distM{\mtrA}{\nnA}{y} \leq
    \distM{\mtrA}{\query}{\nnA} + (5/4) \PDist \leq (1+\eps)
    \distM{\mtrA}{\query}{\nnA}$, assuming $\constB \geq 5$.
\end{proof}

\begin{lemma}%
    \lemlab{c:in}%
    If $\query \in \ballExt{\manifold}{\closest{{\pa}}}{R} \cup
    \ballExt{\manifold}{\closest{{\pb}}}{R}$ then the algorithm
    returns a $(1+\eps)$-\ANN in $\PntSet$ to the query point
    $\query$. 
\end{lemma}

\begin{proof}
    Since the algorithm covered the set
    $\ballExt{\manifold}{\closest{{\pa}}}{R} \cup
    \ballExt{\manifold}{\closest{{\pb}}}{R}$ with a net of radius $r =
    \eps \PDist/\constB$, it follows that
    $\distM{\mtrA}{\query}{\centerset} \leq r$. Let $\cPnt$ be the
    point in the $2$-\ANN search to $\query$ in $\centernettree$. We
    have $\distM{\mtrA}{\query}{\cPnt} \leq 2r$.  Now, the algorithm
    returned the nearest neighbor to $\cPnt$ as the \ANN; that is,
    $\algonnA$ is the nearest neighbor of $\cPnt$ in $\PntSet$.

    Now,
    \begin{align*}
        \distM{\mtrA}{\query}{\algonnA}%
        &\leq%        
        \distM{\mtrA}{\cPnt}{\algonnA} +\distM{\mtrA}{\query}{\cPnt}%
        \leq%
        \distM{\mtrA}{\cPnt}{\algonnA} + 2 r %
        \leq%
        \distM{\mtrA}{\cPnt}{\nnA} + 2 r %
        \\%
        &\leq%
        \distM{\mtrA}{\query}{\nnA} + \distM{\mtrA}{\cPnt}{\query} + 2 r%
        \leq%
        \distM{\mtrA}{\query}{\nnA} + 4 r % 
        =%
        \distM{\mtrA}{\query}{\nnA} + 4 \frac{\eps \PDist}{\constB},%
    \end{align*}
    by the triangle inequality.
    Therefore, if $\distM{\mtrA}{\query}{\algonnA} \geq \PDist / 40$ then,
    \begin{align*}
        \distM{\mtrA}{\query}{\nnA}%
        &\geq%
        \distM{\mtrA}{\query}{\algonnA} - 4 \frac{\eps \PDist}{\constB}%
        \geq%
        (1-\eps/2) \distM{\mtrA}{\query}{\algonnA},
    \end{align*}
    assuming $\constB \geq 320$. Since
    $1/(1-\eps/2) \leq 1+\eps$, we have that
    $\distM{\mtrA}{\query}{\algonnA} \leq (1+\eps)
    \distM{\mtrA}{\query}{\nnA}$.
    
    Similarly, if $\distM{\mtrA}{\query}{\nnA} \geq \PDist / 40$ then,
    \begin{align*}
        \distM{\mtrA}{\query}{\algonnA}%
        &\leq%        
        \distM{\mtrA}{\query}{\nnA} + 4 \frac{\eps \PDist}{\constB}%
        \leq % 
        (1+\eps) \distM{\mtrA}{\query}{\nnA},
    \end{align*}
    assuming $\constB \geq 160$.
    
    We prove by contradiction that the case
    $\distM{\mtrA}{\query}{\nnA} \leq \nfrac{\PDist}{40}$ and
    $\distM{\mtrA}{\query}{\algonnA} \leq \nfrac{\PDist}{40}$ is
    impossible. That is, intuitively, $\PDist$ is roughly the distance
    between $\nnA$ to $\algonnA$, and there is no point that can be
    close to both $\nnA$ and $\algonnA$.  Indeed, under those
    assumptions, $\height{\nnA} \leq \distM{\mtrA}{\query}{\nnA} \leq
    \nfrac{\PDist}{40}$ and $\height{\algonnA} \leq
    \distM{\mtrA}{\query}{\algonnA} \leq \nfrac{\PDist}{40}$.  Observe
    that
    \begin{align*}
        \hMax{\mapped{\sA}}%
        \leq%
        \height{\algonnA} + \diameterX{\mapped{\sA}}%
        \leq%
        \nfrac{\PDist}{40} + \frac{\pDist}{8}%
        \leq%
        \frac{3\PDist}{20}.
    \end{align*}
    and similarly $\hMax{\mapped{\sB}} \leq \nfrac{3\PDist}{20}$.
    This implies that
    \begin{align*}
        (3/4)\pDist%
        &= %
        \pDist\pth{1 - \frac{1}{8} - \frac{1}{8}} \\%
        &\leq%
        \distM{\mtrB}{\mapped{\pa}}{\mapped{\pb}}
        -\diameterX{\mapped{\sA}} -\diameterX{\mapped{\sB}}%
        \leq%
        \distM{\mtrB}{\mapped{\nnA}}{\mapped{\algonnA}}%
        \\%
        &=%
        \cardin{\height{\nnA} - \height{\algonnA}} +
        \distM{\mtrA}{\closest{\nnA}}{\closest{\algonnA}}%
        \leq%
        \PDist/40 + \distM{\mtrA}{\closest{\nnA}}{\nnA} +
        \distM{\mtrA}{\nnA}{\algonnA} +
        \distM{\mtrA}{\algonnA}{\closest{\algonnA}}%
        \\%
        &\leq%
        \PDist/40 + \height{\nnA}%
        + (\distM{\mtrA}{\nnA}{\query} +
        \distM{\mtrA}{\query}{\algonnA}) + \height{\algonnA}%
        \\%
        &\leq%
        \PDist/40 + 3\PDist/20 + \PDist/40 + \PDist/40 + 3 \PDist/20%
        \leq 3 \PDist/8
    \end{align*}
    This implies that $\pDist \leq \PDist / 2$ and thus $\PDist =
    \pDist + \hMax{\mapped{\sA}} + \hMax{\mapped{\sB}} \leq \PDist /2
    + 3\PDist /20 + 3\PDist / 20 = (4/5)\PDist$.  This implies that
    $\PDist \leq 0$. We conclude that
    $\distM{\mtrB}{\mapped{\pa}}{\mapped{\pb}} = \pDist \leq \PDist
    \leq 0$. That implies that $\mapped{\pa} = \mapped{\pb}$, which is
    impossible, as no two points of $\PntSet$ get mapped to the same
    point in $\mtrB$. (And of course, no point can appear in both
    sides of a pair in the \WSPD.)
\end{proof}

The preprocessing time of the above algorithm is dominated by the task
of computing for each point of $\centerset$ its nearest neighbor in
$\PntSet$.  Observe that the algorithm would work even if we only use
$(1+O(\eps))$-\ANN.  Using \thmref{1:epsilon:ANN} to answer these
queries, we get the following result.

\begin{theorem}
    Given a set of $\PntSet \subseteq \mtrA$ of $n$ points, and a
    subspace $\manifold$ of doubling dimension $\dd$, one can
    construct a data-structure requiring space $n
    \eps^{-\porder{\dd}}$, such that given a query point $\query \in
    \manifold$ one can find a $(1 + \eps)$-\ANN to $\query$ in
    $\PntSet$. The query time is $2^{\porder{\dd}} \log (n/\eps)$, and the
    preprocessing time to build this data-structure is $n
    \eps^{-\porder{\dd}} \log n$.
    
    \thmlab{main}
\end{theorem}

% ----------------------------------------------------------------------
% ----------------------------------------------------------------------

\section{Online \ANN}
\seclab{oann}

The algorithms of \secref{epsilon:a:n:n:slow} and \secref{a:n:n}
require that the subspace of the query points is known, in that we can
compute the closest point $\closest{\pntA}$ on $\manifold$ given a
$\pntA \in \mtrA$, and that we can find a net for a ball on
$\manifold$ using $\algNet$, see \subsecref{problem:model}. In this
section we show that if we are able to efficiently answer membership
queries in regions that are the difference of two balls, then we do
not need such explicit access to $\manifold$. We construct an \AVD on
$\manifold$ in an online manner as the query points arrive. When a new
query point arrives, we test for membership among the existing regions
of the \AVD. If a region contains the point we immediately output its
associated \ANN that is already stored with the region. Otherwise we
use an appropriate algorithm to find a nearest neighbor for the query
point and add a new region to the \AVD.

Here we present our algorithm to compute the \AVD in this online
setting and prove that when the query points come from a subspace of
low doubling dimension, the number of regions created is linear.

\subsection{Online \AVD Construction and \ANN Queries}
The algorithm \algBuildAVD{}$(\PntSet,\RegionSet,\query)$ is presented
in \figref{online-ann}. The algorithm maintains a set of regions
$\RegionSet$ that represent the partially constructed \AVD. Given a
query point $\query$ it returns an \ANN from $\PntSet$ and if needed
adds a region $\RegionA_{\query}$ to $\RegionSet$. The quantity
$\diam'$ is a $2$-approximation to the diameter $\diam$ of $\PntSet$,
and can be precomputed in $\order{n}$ time.  Let $\pntA$ be an
arbitrary fixed point of $\PntSet$.
\begin{figure*}
    \centerline{
       \begin{algorithm}
           \algBuildAVD{}$(\PntSet,\RegionSet,\query)$. \+ \+ \\
           \Comment{$\pntA$ is an arbitrary fixed point in $\PntSet$.}\\
           \Comment{$\diam'$ is a $2$-approximation to $\diameter{\PntSet}$.}\\
           \If $\distM{\mtrA}{\query}{\pntA}
           \geq 4\diam'/\eps$ \Then \Return $\pntA$.\\
           \If $\exists \RegionA \in \RegionSet$ with $\query \in
           \RegionA$ \Then
           \+ \\
           \Return the point associated with $\RegionA$.\- \\
           Compute $(1 + \eps/10)$-\ANN $\annA$ of $\query$ in
           $\PntSet$.\\
           $\distA_1 \leftarrow \distM{\mtrA}{\query}{\annA}$. \\
           \If there is no point in $\PntSet \setminus
           \ballExt{\mtrA}{\annA}{\eps \distA_1/3}$
           \Then \+ \\
           $\RegionA_{\query} \leftarrow \ballExt{\mtrA}{\query}{\diam'/4}$. \-\\
           \Else \+ \\
           $\annAf \leftarrow $ furthest point from $\annA$ in
           $\PntSet \cap \ballExt{\mtrA}{\annA}{\eps
              \distA_1/3}$.\\
           $\distB \leftarrow \distM{\mtrA}{\annA}{\annAf}$.
           \;\;\Comment{$\distB \leq \eps \distA_1/3$.}
           \\
           \Comment{One can use any \ANN algorithm, or even brute-force
           to compute $\algoonnB$.} \\ 
           $\algoonnB \leftarrow (1 + \eps/10)$-\ANN of $\query$ in
           $\PntSet \setminus \ballExt{\mtrA}{\annA}{\eps \distA_1/3}$. \\
           $\distA_2 \leftarrow \distM{\mtrA}{\query}{\algoonnB}$. \\
           $\RegionA_{\query} \leftarrow \ballExt{\mtrA}{\query}{\eps
              \distA_2 / 5} \setminus
           \ballExt{\mtrA}{\annA}{5\distB/4\eps}$. \- \\
           Associate $\annA$ with $\RegionA_{\query}$. \\
           $\RegionSet \leftarrow \RegionSet \cup \RegionA_{\query}$.\\
           \Return $\annA$ as the \ANN for $\query$.
       \end{algorithm}%
    }
    \caption{Answering $(1 + \eps)$-\ANN and constructing \AVD.}
    \figlab{online-ann}
    \vspace{-0.6cm}
\end{figure*}

\begin{figure*}[t]
    \bigskip%
    \centerline{
    \begin{tabular}{ccc}
       \includegraphics{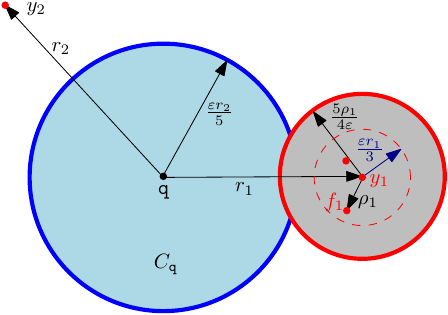}
       &
       \qquad \qquad
       &
       \includegraphics{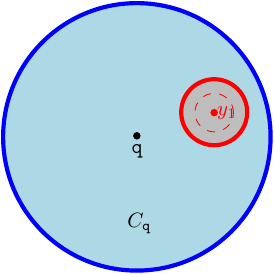}%\\
%       (A)& &(B)
   \end{tabular}%
}
\caption{Examples of a computed \AVD region $\RegionA_{\query}$.}
    \figlab{region}
\end{figure*}

The regions created by the algorithm in \figref{online-ann} are the
difference of two balls. An example region when the balls
$\ballExt{\mtrA}{\query}{\eps \distA_2/5}$ and
$\ballExt{\mtrA}{\annA}{5 \distB/4 \eps}$ intersect is shown in
\figref{region}.  The intuition as to why $\annA$ is a valid \ANN
inside this region is as follows. Since the distance of $\query$ to
$\annA$ is $\distA_1$, the points inside $\ballExt{\mtrA}{\annA}{\eps
   \distA_1/3}$ are all roughly the same distance from $\query$ when
$\query$ is far enough from $\annA$. The
next distance of interest, $\distM{\mtrA}{\query}{\algoonnB} = \distA_2$, 
is the distance to a \ANN of points outside this ball. As
long as we are inside $\ballExt{\mtrA}{\query}{\eps \distA_2/5}$ and far enough
from $\annA$ i.e. $\distM{\mtrA}{\query}{\annA} > 5 \distB / 4 \eps$, the
points outside $\ballExt{\mtrA}{\annA}{\eps \distA_1/3}$ are too far
and cannot be a $(1 + \eps)$-\ANN. But if we get too close to $\annA$
we can no longer be certain that $\annA$ is a valid $(1 + \eps)$-\ANN,
as it is no more true that distances to points inside
$\ballExt{\mtrA}{\annA}{\eps \distA_1/3}$ look all roughly the same.
In other words, there may be points much closer than $\annA$, when we
are close enough to $\annA$. Thus in a small enough neighborhood
around $\annA$ we need to zoom in and possibly create a new region.

\subsection{Correctness}

\begin{lemma}%
    \lemlab{too-far}%
    If $\distM{\mtrA}{\query}{\pntA} \geq 2\diam' + 2\diam'/\eps$ then
    $\pntA$ is a valid $(1+\eps)$-\ANN.
\end{lemma}

\begin{proof}
    Since $\diam'$ is a $2$-approximation to the diameter of
    $\PntSet$, so $2 \diam' \geq \diam = \diameterX{\PntSet}$. This
    means $\distM{\mtrA}{\query}{\pntA} \geq \diam + \diam/\eps$. Let
    $\nnA \in \PntSet$ be the closest point to $\query$. By the
    triangle inequality,
    \begin{equation*}
        \diam + \diam / \eps \leq \distM{\mtrA}{\query}{\pntA} 
        \leq \distM{\mtrA}{\query}{\nnA} + \distM{\mtrA}{\nnA}{\pntA} 
        \leq \distM{\mtrA}{\query}{\nnA} + \diam.
    \end{equation*}
    As such $\diam \leq \eps \distM{\mtrA}{\query}{\nnA}$.  We
    conclude $\distM{\mtrA}{\query}{\pntA} \leq
    \distM{\mtrA}{\query}{\nnA} + \distM{\mtrA}{\nnA}{\pntA} \leq (1 +
    \eps) \distM{\mtrA}{\query}{\nnA}$.
\end{proof}

\begin{lemma}%
    \lemlab{no-region}%
    If there is no region in $\RegionSet$ containing $\query$ then the
    algorithm outputs a valid $(1 + \eps/10)$-\ANN.
\end{lemma}

\begin{proof}
    We output $\annA$ which is a $(1 + \eps/10)$-\ANN of $\query$.
\end{proof}

\begin{lemma}%
    \lemlab{region-good}%
    The $(1 + \eps/10)$-\ANN $\annA$ found in the algorithm is a $(1 +
    \eps)$-\ANN for any point $\queryB \in \RegionA_{\query}$.
\end{lemma}

\begin{proof}
    Let $\distA_1 = \distM{\mtrA}{\query}{\annA}$ and $\distA_2 =
    \distM{\mtrA}{\query}{\algoonnB}$. There are two possibilities.
    
    If the region $\RegionA_{\query}$ is the ball
    $\ballExt{\mtrA}{\query}{\diam'/4}$ constructed when there is no
    point in $\PntSet \setminus \ballExt{\mtrA}{\annA}{\eps \distA_1 /
       3}$, then $\diam = \diameterX{\PntSet} \leq 2 \eps \distA_1 /
    3$.  As such,
    \begin{equation*}
        \distM{\mtrA}{\query}{\PntSet} \geq 
        \frac{\distM{\mtrA}{\query}{\annA}}{1+\eps/10} = 
        \frac{\distA_1}{1 + \eps/10}
        \geq \frac{3\diam}{2\eps(1 + \eps/10)}
        = \frac{3\diam}{2\eps + \eps^2/5} 
        \geq (4/3) \frac{\diam}{\eps}.
    \end{equation*}
    It is not hard to see that in this case, $\annA$ is a valid $(1 +
    \eps)$-\ANN for any point inside
    $\ballExt{\mtrA}{\query}{\diam'/4} \subseteq
    \ballExt{\mtrA}{\query}{\diam/4}$, as
    $\distM{\mtrA}{\ballExt{\mtrA}{\query}{\diam/4}}{\PntSet} \geq
    \diam/\eps$, for $\eps$ sufficiently small.

    \smallskip
    
    Otherwise, if the set $\PntSet \setminus
    \ballExt{\mtrA}{\annA}{\eps \distA_1/3}$ is nonempty then let
    $\algoonnB$ be a $(1 + \eps/10)$-\ANN of $\query$ in $\PntSet
    \setminus \ballExt{\mtrA}{\annA}{\eps \distA_1/3}$ and let
    $\distA_2 = \distM{\mtrA}{\query}{\algoonnB}$. We break the
    analysis into two cases.
    
    \begin{enumerate}[(i)]
        \item If $\distA_2 \leq 2 \distA_1$, then let $\queryB$ be any
        point in $\RegionA_{\query}$ and let $\nnB \in \PntSet$ be its
        nearest neighbor. If $\nnB = \annA$ there is nothing to
        show. Otherwise $\distM{\mtrA}{\query}{\queryB} \leq \eps
        \distA_2/5$ and by the triangle inequality we have
        \begin{align*}
            \distM{\mtrA}{\queryB}{\nnB} %
            &\geq% 
            \distM{\mtrA}{\query}{\nnB} - \distM{\mtrA}{\query}{\queryB}%
            \geq \distM{\mtrA}{\query}{\nnB} - \eps \distA_2 / 5 \\%
            &\geq%
            \distM{\mtrA}{\query}{\annA}/(1 + \eps/10) -
            \eps 2 \distA_1 / 5 \\
            &\geq%
            (1 - \eps/2) \distA_1,
        \end{align*}
        as $\distM{\mtrA}{\query}{\nnB} \geq
        \distM{\mtrA}{\query}{\PntSet} \geq
        \distM{\mtrA}{\query}{\annA}/(1 + \eps/10)$ and $\distA_1 =
        \distM{\mtrA}{\query}{\annA}$.  Again, by the triangle
        inequality and the above, we have
        \begin{align*}
            \distM{\mtrA}{\queryB}{\annA} &\leq%
            \distM{\mtrA}{\query}{\annA} +
            \distM{\mtrA}{\query}{\queryB}%
            \leq%
            \distM{\mtrA}{\query}{\annA} + 2 \eps \distA_1 / 5%
            = (1 + 2 \eps/5) \distA_1%
            \\%
            &\leq \frac{1 + 2 \eps/5}{1 - \eps/2}
            \distM{\mtrA}{\queryB}{\nnB}%
            \leq (1 + \eps)\distM{\mtrA}{\queryB}{\nnB},
        \end{align*}
        for $\eps \leq 1/5$.
        
        \item If $\distA_2 > 2 \distA_1$ then let $\annAf$ be the
        furthest point from $\annA$ inside
        $\ballExt{\mtrA}{\annA}{\eps \distA_1/3}$ and let $\distB =
        \distM{\mtrA}{\annA}{\annAf}$. Let $\queryB$ be any point in
        $\RegionA_{\query}$ and as before let $\nnB \in \PntSet$ be
        its nearest neighbor.  We claim that the nearest neighbor of
        $\queryB$ in $\PntSet$ lies in
        $\ballExt{\mtrA}{\annA}{\distB}$. To see this, let $\pntD$ be
        any point in $\PntSet \setminus
        \ballExt{\mtrA}{\annA}{\distB}$. Noting that the distance from
        $\query$ to the closest point in $\PntSet$ outside
        $\ballExt{\mtrA}{\annA}{\distB}$ is at least $\distA_2/(1 +
        \eps/10)$ and by triangle inequality we have,
        \begin{align*}
            \distM{\mtrA}{\queryB}{\pntD}%
            &\geq%
            \distM{\mtrA}{\query}{\pntD} -
            \distM{\mtrA}{\query}{\queryB}%
            \geq \distM{\mtrA}{\query}{\pntD} - \eps \distA_2/5%
            \\%
            &\geq%
            \distA_2/(1 + \eps/10) - \eps \distA_2/5%
            > (1 - 3 \eps / 10) \distA_2.
        \end{align*}
        
        On the other hand, as $\distA_1 =
        \distM{\mtrA}{\query}{\annA}$ and $\distA_1 < \distA_2 / 2$,
        we have
        \begin{align*}
            \distM{\mtrA}{\queryB}{\annA}%
            &\leq%
            \distM{\mtrA}{\query}{\annA} +
            \distM{\mtrA}{\query}{\queryB}%
            \leq%
            \distM{\mtrA}{\query}{\annA} + \eps \distA_2 / 5%
            =%
            \distA_1 + \eps \distA_2 / 5%
            <%
            \distA_2 / 2 + \eps \distA_2 / 5%
            \\ &%
            \leq%
            (1 - 3\eps/10) \distA_2%
            <%
            \distM{\mtrA}{\queryB}{\pntD},
        \end{align*}
        by the above. As such, no point in $\PntSet \setminus
        \ballExt{\mtrA}{\annA}{\distB}$ can be the nearest neighbor of
        $\queryB$ for $\eps < 1$. As such $\nnB \in
        \ballExt{\mtrA}{\annA}{\distB}$. Now,
        \begin{equation}%
            \eqnlab{dist:u:b}%
            \distM{\mtrA}{\queryB}{\annA}
            \leq \distM{\mtrA}{\queryB}{\nnB} +
            \distM{\mtrA}{\nnB}{\annA} \leq
            \distM{\mtrA}{\queryB}{\nnB} + \distB.
        \end{equation}
        Now $\queryB \in \RegionA_{\query} =
        \ballExt{\mtrA}{\query}{\eps \distA_2 / 5} \setminus
        \ballExt{\mtrA}{\annA}{5 \distB /4 \eps}$, and thus
        $\distM{\mtrA}{\queryB}{\annA} > 5 \distB /4 \eps$.  Thus,
        \begin{equation}%
            \eqnlab{dist:l:b}%
            \distM{\mtrA}{\queryB}{\nnB}
            \geq%
            \distM{\mtrA}{\queryB}{\annA} -
            \distM{\mtrA}{\annA}{\nnB} \geq
            \distM{\mtrA}{\queryB}{\annA} - \distB%
            \geq%
            \pth{\frac{5}{4 \eps} - 1} \distB.
        \end{equation}
        Therefore from \eqnref{dist:u:b} and \eqnref{dist:l:b}, we have
        \begin{align*}
            \distM{\mtrA}{\queryB}{\annA}%
            &\leq%
            \distM{\mtrA}{\queryB}{\nnB} + \distB %
            \leq%
            \pth{\MakeSBig 1 + \frac{1}{5/4\eps - 1}}
            \distM{\mtrA}{\queryB}{\nnB}%
            =%
            \pth{\MakeSBig 1 + \frac{4\eps}{5 - 4\eps}}
            \distM{\mtrA}{\queryB}{\nnB}\\
            &\leq%
            (1 + \eps) \distM{\mtrA}{\queryB}{\nnB}.
        \end{align*}
        for $\eps \leq 1/4$.
    \end{enumerate}
\end{proof}

\subsection{Bounding the number of regions created }

The online algorithm presented in \figref{online-ann} is valid for any
general metric space $\mtrA$, without any restriction on the subspace
of query points. However, when the query points are restricted to lie
in a subspace $\manifold$ of low doubling dimension $\dd$, then one
can show that at most $n \eps^{-\porder{\dd}}$ regions are created
overall, where $n = \cardin{\PntSet}$. There are two types of regions
created. The \emphi{outer} regions are created when $\PntSet \setminus
\ballExt{\mtrA}{\annA}{\eps \distA_1 / 3}$ is empty and the
\emphi{inner} regions are created when this condition does not
hold. An example of an inner region is shown in \figref{region}.

\subsubsection{Bounding the number of outer regions}

First we show that there are at most $\eps^{-\porder{\dd}}$ outer
regions created.

\begin{lemma}%
    \lemlab{ftrub}%
    When all the queries to the algorithm come from a subspace of
    doubling dimension $\dd$, then at most $\eps^{-\porder{\dd}}$
    outer regions 
    are created overall.
\end{lemma}

\begin{proof}
    Any two query points creating distinct outer regions occur at a
    distance of at least $\diam'/4$ from each other. However all of
    them occur inside a ball of radius $4\diam'/\eps$ around $\pntA$.
    Thus the spread of the set containing all these query points is
    bounded by $ \pth{4\diam'/\eps} \,/ \, \pth[]{\diam'/4} =
    O(1/\eps)$. As such, there are at most $\eps^{-\porder{\dd}}$ such
    points.
\end{proof}

\subsubsection{Bounding the number of inner regions}

We now consider the inner regions created by the algorithm.  Consider
the mapped point set $\mapped{\PntSet}$ in the space $\mtrB$, see
\secref{embedding}. Fix a $c$-\WSPD $\brc{ \MakeSBig\!
   \brc{\mapped{A_1},\mapped{B_1}},\dots,\brc{\mapped{A_s},\mapped{B_s}}}$
of $\mapped{\PntSet}$ where $c$ is a constant to be specified shortly
and $s = c^{\porder{\dd}} n$ is the number of pairs. Let 
$A_i, B_i \subseteq \PntSet$ denote the corresponding ``unmapped'' points
corresponding to $\mapped{A_i}, \mapped{B_i}$, that is,
$A_i = \brc{\pntA \in \PntSet ~ | ~ \mapped{\pntA} \in \mapped{A_i}}$
and $B_i = \brc{\pntA \in \PntSet ~ | ~ \mapped{\pntA} \in \mapped{B_i}}$.
If a query point
$\query$ creates a new inner region we shall assign it to a set
$\QuerySetA_i$ associated with the pair
$\brc{\mapped{A_i},\mapped{B_i}}$, if the pair of points
$\mapped{\annA}, \mapped{\algoonnB}$ of the algorithm satisfy
$\mapped{\annA} \in \mapped{A_i}$ and $\mapped{\algoonnB} \in
\mapped{B_i}$.  Similarly assign $\query$ to the set $\QuerySetB_i$ if
$\mapped{\annA} \in \mapped{B_i}$ and $\mapped{\algoonnB} \in
\mapped{A_i}$.

Thus, the query points that gave rise to new regions are now
associated with pairs of the \WSPD.  Our analysis bounds the size of
the sets $\QuerySetA_i$ and $\QuerySetB_i$ associated with a pair
$\brc{\mapped{A_i},\mapped{B_i}}$, for $i=1,\ldots, s$, thus bounding
the total number of regions created.

Let $\mapped{\QuerySetA_i} = \brc{\mapped{\query} \sep{\query \in
      \QuerySetA_i}} \subseteq \mtrB $ and $\mapped{\QuerySetB_i} =
\brc{\mapped{\query} \sep{\query \in \QuerySetB_i}}$, for $i=1,\ldots,
s$.  For a pair $\brc{\mapped{A_i},\mapped{B_i}}$ of the \WSPD we
define the numbers $\hMax{\mapped{A_i}} = \max_{(\pntC,h) \in
   \mapped{A_i}} h$. Similarly let $\hMax{\mapped{B_i}} =
\max_{(\pntD,h) \in \mapped{B_i}} h$. Also, let
\begin{align*}
    \lDist{i} = \max_{\mapped{\pntC} \in \mapped{A_i}, \mapped{\pntD}
       \in \mapped{B_i} }
    \distM{\mtrA}{\closest{\pntC}}{\closest{\pntD}}
    \qquad \text{ and } \qquad %
    \LDist{i} = \lDist{i} + \hMax{\mapped{A_i}} + \hMax{\mapped{B_i}}.
\end{align*}

The following sequence of lemmas establish our claim. The basic
strategy is to show that the set $\mapped{\QuerySetA_i}$ has spread
$\order{1/\eps^2}$. This holds analogously for $\mapped{\QuerySetB_i}$
and so we will only work with $\mapped{\QuerySetA_i}$. We will assume
that $c$ is a sufficiently large constant and $\eps$ is sufficiently
small.

\begin{lemma}%
    \lemlab{good:diam}%
    For any $i$, we have $\diameterY{\mtrB}{\mapped{A_i}} \leq
    \LDist{i}/c $ and $\diameterY{\mtrB}{\mapped{B_i}} \leq
    \LDist{i}/c$.
\end{lemma}

\begin{proof}
    By the construction of the \WSPD, we have that
    $\diameterY{\mtrB}{\mapped{A_i}} \leq \distM{\mtrB}{\mapped{A_i}}
    {\mapped{B_i}}/c$. Moreover, we have
    \begin{align*}
        \distM{\mtrB}{\mapped{A_i}} {\mapped{B_i}}%
        &=%
        \min_{\pntA' \in \mapped{A_i}, \pntB' \in \mapped{B_i}}
        \distM{\mtrB}{\pntA'}{\pntB'}%
        =%
        \min_{\pntA' \in \mapped{A_i}, \pntB' \in \mapped{B_i}} \pth{
           \MakeBig \distM{\mtrA}{\MakeSBig
              \closest{\pntA}}{\closest{\pntB}} + \cardin{
              \height{\pntA} - \height{\pntB} }}
        \\
        &\leq %
        \,\lDist{i} + \min_{\pntA' \in \mapped{A_i}, \pntB' \in
           \mapped{B_i}} \pth{\MakeBig \cardin{ \height{\pntA}} +
           \cardin{ \height{\pntB} }}%
        \leq %
        \lDist{i} + \hMax{\mapped{A_i}} + \hMax{\mapped{B_i}} =
        \LDist{i}. 
    \end{align*}
    This implies that $\diameterY{\mtrB}{\mapped{A_i}} \leq
    \LDist{i}/c$, and similarly $\diameterY{\mtrB}{\mapped{B_i}} \leq
    \LDist{i}/c$.
\end{proof}

\begin{lemma}%
    \lemlab{diamub}%
    We have $\diameter{\mapped{\QuerySetA_i}} =
    \order{\LDist{i}/\eps}$.
\end{lemma}
\begin{proof}
    Let $\query$ be a (query) point in $\QuerySetA_i$. By assumption
    we have $\mapped{\annA} \in \mapped{A_i}$ and $\mapped{\algoonnB}
    \in \mapped{B_i}$.  By the triangle inequality,
    \begin{align*}
        \distM{\mtrA}{\annA}{\algoonnB}%
        &\leq%
        \distM{\mtrA}{\annA}{\closest{\annA}} +
        \distM{\mtrA}{\closest{\annA}}{\closest{\algoonnB}}+
        \distM{\mtrA}{\closest{\algoonnB}}{\algoonnB} %
        \leq%
        \hMax{\mapped{A_i}} + \lDist{i} + \hMax{\mapped{B_i}} \\
        &\leq%
        \LDist{i}.
    \end{align*}
    On the other hand, since the point $\algoonnB$ is outside
    $\ballExt{\mtrA}{\annA}{\eps \distA_1 /3}$, we have that
    $\distM{\mtrA}{\annA}{\algoonnB} > \eps \distA_1/3$, where
    $\distA_1 = \distM{\mtrA}{\query}{\annA}$.  This gives us
    $\distA_1 < (3/\eps) \distM{\mtrA}{\annA}{\algoonnB} < 3\LDist{i}/
    \eps$.  By \lemref{embed:low:d:d},
    $\distM{\mtrB}{\mapped{\query}}{\mapped{\annA}} \leq 3
    \distM{\mtrA}{\query}{\annA} = 3 \distA_1 < 9\LDist{i}/\eps$.
    Also, we have,
    \begin{align}
        \distM{\mtrB}{\mapped{\annA}}{\mapped{\algoonnB}} 
        &=%
        \distM{\mtrA}{\closest{\annA}}{\closest{\algoonnB}} + 
        \abs{\height{\annA} - \height{\algoonnB}} %
        \leq%
        \lDist{i} + \hMax{\mapped{A_i}} + \hMax{\mapped{B_i}} =
        \LDist{i}.
        \eqnlab{L:i:bounds}
    \end{align}
    Let $\queryB$ be any other point in $\QuerySetA_i$, and let the
    points $\annAB$ and $\annBB$ be the points found by the algorithm
    such that $\mapped{\annAB} \in \mapped{A_i}$ and $\mapped{\annBB}
    \in \mapped{B_i}$. Since $\mapped{\annA}$ is also in
    $\mapped{A_i}$, we have by \lemref{good:diam} that
    $\distM{\mtrB}{\mapped{\annA}}{\mapped{\annAB}} \leq
    \diameterY{\mtrB}{A_i} \leq \LDist{i} / c$. As such,
    \begin{align*}
        \diameter{\mapped{\QuerySetA_i}} &=
        \max\limits_{\mapped{\query},\mapped{\queryB} \in
           \mapped{\QuerySetA_i}}
        \distM{\mtrB}{\mapped{\query}}{\mapped{\queryB}} \leq
        \max\limits_{\mapped{\query},\mapped{\queryB} \in
           \mapped{\QuerySetA_i}}
        (\distM{\mtrB}{\mapped{\query}}{\mapped{\annA}} +
        \distM{\mtrB}{\mapped{\annA}}{\mapped{\annAB}}
        + \distM{\mtrB}{\mapped{\queryB}}{\mapped{\annAB}}) \\
        &\leq 9\LDist{i}/\eps + \LDist{i}/c + 9 \LDist{i} / \eps =
        \order{\LDist{i}/\eps},
    \end{align*}
    for $\eps$ small enough.
\end{proof}

\begin{lemma}%
    \lemlab{distA2lb}%
    For a query point $\query$, the associated distances $\distA_2$
    and $\LDist{i}$ satisfy $\distA_2 \geq \LDist{i}/18$.    
\end{lemma}
\begin{proof}
    Let $\mapped{\pntC}$ be the point with maximum height in
    $\mapped{A_i}$; that is $\height{\pntC} = \hMax{\mapped{A_i}}$.
    By \lemref{good:diam}, we have
    $\distM{\mtrB}{\mapped{\pntC}}{\mapped{\annA}} \leq \LDist{i}/c$.
    The definition of the distance in $\mtrB$, gives
    \begin{align*}
        \hMax{\mapped{A_i}} - \height{\annA}%
        \leq%
        \cardin{\hMax{\mapped{A_i}} - \height{\annA} }%
        =%
        \cardin{\height{\pntC} - \height{\annA}}%
        \leq%
        \distM{\mtrB}{\mapped{\pntC}}{\mapped{\annA}}%
        \leq%
        \LDist{i}/c,
    \end{align*}
    and so $\height{\annA} \geq \hMax{\mapped{A_i}} - \LDist{i}/c$.
    Similarly we have, $\height{\algoonnB} \geq \hMax{\mapped{B_i}} -
    \LDist{i}/c$.  We have $\distA_1 = \distM{\mtrA}{\query}{\annA}
    \geq \distM{\mtrA}{\annA}{\manifold} =
    \distM{\mtrA}{\annA}{\closest{\annA}} = \height{\annA}$ and
    similarly $\distA_2 = \distM{\mtrA}{\query}{\algoonnB} \geq
    \height{\algoonnB}$. Noting that, $\distA_2 \geq
    \distM{\mtrA}{\query}{\PntSet} \geq \distA_1/(1 + \eps/10) \geq
    (10/11)\distA_1$ we get,
    \begin{equation} 
        \eqnlab{h:i:n} %
        2.1 \distA_2 = \distA_2 +
        \frac{11}{10} \distA_2 \geq \distA_2 + \distA_1 \geq
        \height{\algoonnB} + \height{\annA} \geq \hMax{\mapped{A_i}} +
        \hMax{\mapped{B_i}} - \frac{2\LDist{i}}{c}.
    \end{equation}
    
    Let $\mapped{\pntD} \in \mapped{A_i}$ and $\mapped{\pntE} \in
    \mapped{B_i}$ be such that
    $\distM{\mtrA}{\closest{\pntD}}{\closest{\pntE}} = \lDist{i}$.
    Observing that $\distM{\mtrB}{\mapped{\query}}{\mapped{\annA}}
    \leq 3 \distM{\mtrA}{\query}{\annA} = 3 \distA_1$ and similarly
    $\distM{\mtrB}{\mapped{\query}}{\mapped{\algoonnB}} \leq 3
    \distM{\mtrA}{\query}{\algoonnB} = 3 \distA_2$, we have by the
    triangle inequality that
    \begin{align*}
        &%
        \distM{\mtrB}{\mapped{\query}}{\mapped{\pntD}} \leq
        \distM{\mtrB}{\mapped{\query}}{\mapped{\annA}} +
        \distM{\mtrB}{\mapped{\annA}}{\mapped{\pntD}} \leq 3 \distA_1
        + \diameterX{\mapped{A_i}} \leq 3 \distA_1 + \LDist{i}/c ,
        \\%
        \text{ and } \qquad &%
        \distM{\mtrB}{\mapped{\query}}{\mapped{\pntE}} \leq
        \distM{\mtrB}{\mapped{\query}}{\mapped{\algoonnB}} +
        \distM{\mtrB}{\mapped{\algoonnB}}{\mapped{\pntE}} \leq 3
        \distA_2 + \diameterX{\mapped{B_i}} \leq 3 \distA_2 +
        \LDist{i}/c,
    \end{align*}
    by \lemref{good:diam}.  By the triangle inequality, we have
    \begin{align*}
        \lDist{i}%
        &\leq%
        \distM{\mtrB}{\mapped{\pntD}}{\mapped{\pntE}} %
        \leq%
        \distM{\mtrB}{\mapped{\pntD}}{\mapped{\query}} +
        \distM{\mtrB}{\mapped{\query}}{\mapped{\pntE}} %
        \leq%
        3 \distA_1 + 3 \distA_2 + \frac{2\LDist{i}}{c}
        \leq%
        6.3 \distA_2 + \frac{2\LDist{i}}{c},
    \end{align*}
    as $\distA_1 \leq (11/10) \distA_2$.  Thus we have,
    \begin{equation}  
        6.3 \distA_2 \geq \lDist{i} -
        \frac{2\LDist{i}}{c}.
        \eqnlab{l:i:n:e:q}
    \end{equation}
    By \Eqref{h:i:n} and \Eqref{l:i:n:e:q}, we have, for $c \geq 8$,
    that
    \begin{align*}
        9 \distA_2 \geq 2.1 \distA_2 + 6.3 \distA_2%
        &\geq%
        \pth{ \hMax{\mapped{A_i}}+\hMax{\mapped{B_i}} - \frac{2 \LDist{i}}{c} }
        + \pth{ \lDist{i}-\frac{2\LDist{i}}{c} } \\
        &=%
        \hMax{\mapped{A_i}} + \hMax{\mapped{B_i}} + \lDist{i} - \frac{4 \LDist{i}}{c}
        \geq \LDist{i} - \frac{\LDist{i}}{2}%
        =%
        \frac{\LDist{i}}{2},
    \end{align*}
    which implies $\LDist{i} \leq 18 \distA_2$.
\end{proof}

Suppose $\queryB$ was added to $\QuerySetA_i$ after $\query$. We want
to show that for $\query, \queryB \in \QuerySetA_i$ we must have
$\distM{\mtrB}{\mapped{\query}}{\mapped{\queryB}} > \eps \distA_2 / 5$
where $\distA_2 = \distM{\mtrA}{\query}{\algoonnB}$. We establish this
through a sequence of lemmas. The proof is essentially by
contradiction, and the next four lemmas assume the contrary to derive
a contradiction. Roughly speaking, the assumption that
$\distM{\mtrB}{\mapped{\query}}{\mapped{\queryB}} = 
 \distM{\mtrA}{\query}{\queryB} \leq \eps \distA_2 / 5$
places $\queryB$ in the chipped off region of the crescent 
region $\RegionA_{\query}$. It turns out that $\queryB$ is
far from both the approximate nearest neighbor of $\query$,
which is $\annA$ and the approximate nearest neighbor
of $\query$ outside an environ of $\annA$, which is $\algoonnB$.
Under the assumption $\query, \queryB \in \QuerySetA_i$ we should
however be able to find the corresponding approximate nearest 
neighbors for $\queryB$ close to those of $\query$. Enforcing the
constraint that the approximate nearest neighbor of $\queryB$
cannot be the second approximate nearest neighbor of $\query$, which
is $\algoonnB$, leads to either counting discrepancies or geometric
contradictions arising from the triangle inequality.

\begin{lemma}%
    \lemlab{qswsp1}%
    Let $\query, \queryB$ be two points of $\QuerySetA_i$, such that
    $\queryB$ was added after $\query$. If
    $\distM{\mtrA}{\query}{\queryB} \leq \eps \distA_2 / 5$, then
    \begin{inparaenum}[(i)]
        \item
        $\distM{\mtrA}{\queryB}{\annA} \leq (5 / 4 \eps)
        \distB$, and
        \item    $\distA_2 \geq (2 / \eps) \distA_1$.
    \end{inparaenum}
\end{lemma}

\begin{proof}
    Since $\queryB$ created a new region it lies outside
    $\RegionA_{\query} = \ballExt{\mtrA}{\query}{\eps \distA_2/ 5}
    \setminus \ballExt{\mtrA}{\annA}{5 \distB / 4 \eps}$.  Since by
    assumption $\queryB \in \ballExt{\mtrA}{\query}{\eps \distA_2 /
       5}$, it must be the case that $\queryB \in
    \ballExt{\mtrA}{\annA}{5 \distB / 4 \eps}$, as otherwise $\queryB
    \in \RegionA_\query$. Thus, these two balls intersect, and
    \begin{align*}
        \frac{\eps}{5} \distA_2 + \frac{5}{4 \eps} \distB %
        \geq%
        \distM{\mtrA}{\query}{\annA} = \distA_1.
    \end{align*}
    But $\distB \leq \eps \distA_1 / 3$ and so $\distA_1 \geq (3 /
    \eps)\distB$, implying
    \begin{align*}
        \frac{\eps}{5} \distA_2 + \frac{5}{12} \distA_1 %
        \geq%
        \frac{\eps}{5} \distA_2 + \frac{5}{12} \cdot \frac{3}{\eps}
        \distB %
        =%
        \frac{\eps}{5} \distA_2 + \frac{5}{4 \eps} \distB %
        \geq%
        \distA_1
        \quad \implies \quad%
        \distA_2 \geq \frac{35}{12 \eps} \distA_1%
        \geq%
        \frac{2}{\eps} \distA_1.
    \end{align*}
    \aftermathA
\end{proof}

\begin{lemma}%
    \lemlab{qswsp2}%
    Let $\query, \queryB$ be two points in $\QuerySetA_i$ such that 
    $\queryB$ was added
    after $\query$. If $\distM{\mtrA}{\query}{\queryB} \leq \eps
    \distA_2 / 5$ then, for sufficiently small $\eps$ and sufficiently
    large $c$, we have that
    \begin{compactenum}[\rm\qquad(A)]
        \item \lempntlab{q:s:A}%
        $\distA_1 \leq \eps \LDist{i}$.

        \item \lempntlab{q:s:B}%
        $\distM{\mtrA}{\annA}{\queryB} \leq 5 \distA_1 / 12 \leq \eps
        \LDist{i}$.

        \item \lempntlab{q:s:C}%
        $\distM{\mtrA}{\queryB}{B_i} \geq \LDist{i}/120$.
    \end{compactenum}
\end{lemma}

\begin{proof}
    \begin{inparaenum}[(A)]
        \item By \Eqref{L:i:bounds} we have
        $\distM{\mtrB}{\mapped{\annA}}{\mapped{\algoonnB}} \leq
        \LDist{i}$.  Now, by \lemref{qswsp1}, we have $\distA_2 \geq
        (2 / \eps) \distA_1$. As such, by the triangle inequality, and
        by \lemref{embed:low:d:d}, we have
        \begin{align}
            \eqlab{above}%
            \LDist{i}%
            &\geq %
            \distM{\mtrB}{\mapped{\annA}}{\mapped{\algoonnB}} %
            \geq%
            \distM{\mtrB}{\mapped{\query}}{\mapped{\algoonnB}} -
            \distM{\mtrB}{\mapped{\query}}{\mapped{\annA}} %
            \geq \distM{\mtrA}{\query}{\algoonnB} -
            3 \distM{\mtrA}{\query}{\annA} \\
            \nonumber &\geq \distA_2 - 3 \distA_1 \geq 2\distA_1/\eps
            - 3 \distA_1 \geq \distA_1 / \eps,
        \end{align}
        for $\eps \leq 1/3$. Thus $\LDist{i} \geq \distA_1 / \eps$.
        
        \item In terms of $\distA_2$, by \Eqref{above}, we have
        \begin{align}
            \distM{\mtrB}{\mapped{\annA}}{\mapped{\algoonnB}} &\geq%
            \distA_2 - 3 \distA_1%
            \geq%
            \distA_2 - \frac{3 \eps \distA_2}{2}%
            \geq%
            \frac{\distA_2}{2}%
            \geq%
            \frac{\LDist{i}}{36},%
            \eqlab{a:b}
        \end{align}
        since by \lemref{distA2lb} $\distA_2 \geq \LDist{i}/{18}$ for
        $\eps \leq 1/3$ and by \lemref{qswsp1} $\distA_1 \leq \eps
        \distA_2/2$.  Now $\queryB$ lies inside
        $\ballExt{\mtrA}{\annA}{(5 / 4 \eps) \distB}$ and as $\distB
        \leq (\eps /3 ) \distA_1$ (see \figref{region}), we have
        \begin{align*}
            \distM{\mtrA}{\annA}{\queryB} \leq (5 / 4 \eps) \distB
            \leq (5 / 4 \eps) (\eps / 3 ) \distA_1 \leq 5 \distA_1 /
            12 \leq \distA_1 \leq \eps \LDist{i},
        \end{align*}
        by (A).
        
        \item Let $\pntD$ be an arbitrary point in $B_i$ and notice
        that by \Eqref{a:b} and the triangle inequality we have,
        \begin{align*}
            \distM{\mtrB}{\mapped{\queryB}}{\mapped{\pntD}}%
            &\geq
            \distM{\mtrB}{\mapped{\annA}}{\mapped{\pntD}} -
            \distM{\mtrB}{\mapped{\queryB}}{\mapped{\annA}}
            \geq \distM{\mtrB}{\mapped{\annA}}{\mapped{\algoonnB}} -
            \distM{\mtrB}{\mapped{\algoonnB}}{\mapped{\pntD}}
            - \distM{\mtrB}{\mapped{\queryB}}{\mapped{\annA}} \\
            &\geq \frac{\LDist{i}}{36} - \diameterX{\mapped{B_i}} -
            3\distM{\mtrA}{\queryB}{\annA}%
            \geq%
            \frac{\LDist{i}}{36} -
            \frac{\LDist{i}}{c} - 3\eps \LDist{i}
            \geq \frac{\LDist{i}}{40},
        \end{align*}
        by \lemref{embed:low:d:d} \lempntref{embed:low:d:d:B} and
        \lemref{good:diam} for
        sufficiently small $\eps$ and sufficiently large $c$. Thus,
        \lemref{embed:low:d:d} \lempntref{embed:low:d:d:B}, implies
        that $\distM{\mtrA}{\queryB}{\pntD} \geq
        \distM{\mtrB}{\mapped{\queryB}}{\mapped{\pntD}}/3 \geq
        \LDist{i}/120$.
    \end{inparaenum}
\end{proof}

\begin{lemma}%
    \lemlab{qswsp3}%
    Let $\query, \queryB$ be two points in $\QuerySetA_i$ such that
    $\queryB$ was added after $\query$, and suppose
    $\distM{\mtrA}{\query}{\queryB} \leq \eps \distA_2 / 5$.  Let
    $\modA_i = A_i \cup \brc{\annAf}$, where $\annAf$ is the furthest
    point from $\annA$ in the set $\ballExt{\mtrA}{\annA}{\eps
       \distA_1 / 3} \cap \PntSet$. Then, for sufficiently small
    $\eps$ and sufficiently large $c$, we have $B_i \cap \modA_i =
    \emptyset$.  In particular, we have $\distM{\mtrA}{\queryB}{B_i} >
    2 \max_{\pntC \in \modA_i} \distM{\mtrA}{\queryB}{\pntC}$.
\end{lemma}

\begin{proof}
    First, let $\pntC$ be any point in $A_i$. Then, by \lemref{embed:low:d:d}
    \lempntref{embed:low:d:d:B}, the triangle inequality, \lemref{good:diam} 
    and
    \lemref{qswsp2} we have, for $c$ sufficiently large and $\eps$
    sufficiently small, that
    \begin{align}
        \distM{\mtrA}{\queryB}{\pntC}%
        &\leq%
        \distM{\mtrB}{\mapped{\queryB}}{\mapped{\pntC}}%
        \leq%        
        \distM{\mtrB}{\mapped{\queryB}}{\mapped{\annA}} +
        \distM{\mtrB}{\mapped{\annA}}{\mapped{\pntC}}
        \leq%
        3\distM{\mtrA}{\queryB}{\annA} 
        + \diameterY{\mtrB}{\mapped{A_i}}
        \nonumber \\%
        &%
        \leq%
        3 \eps \LDist{i} + \frac{\LDist{i}}{c} %
        < %
        \frac{\LDist{i}}{240}. \nonumber
    \end{align}
    We also have by the triangle inequality,
    \begin{align*}
        \distM{\mtrA}{\queryB}{\annAf}%
        \leq%
        \distM{\mtrA}{\queryB}{\annA} + \distM{\mtrA}{\annA}{\annAf}
        \leq%
        \eps \LDist{i} + \frac{\eps}{3} \distA_1 %
        \leq%
        \eps\LDist{i} + \frac{\eps^2}{3} \LDist{i} %
        <%
        \frac{\LDist{i}}{240},
    \end{align*}
    since $\distM{\mtrA}{\annA}{\annAf} \leq \eps \distA_1 /3 $ and by
    \lemref{qswsp2}.  As such, for sufficiently large $c$ and
    small $\eps$, we have
    \begin{align}
        \eqlab{q:2:u} \max_{\pntC \in \modA_i}
        \distM{\mtrA}{\queryB}{\pntC} < \frac{\LDist{i}}{240}.
    \end{align}
    On the other hand, for any $\pntD \in B_i$, we have by
    \lemref{qswsp2} \lempntref{q:s:C} that
    $\distM{\mtrA}{\queryB}{\pntD} \geq \LDist{i}/120$. As such, by
    \Eqref{q:2:u}, we have
    \begin{align*}
        \distM{\mtrA}{\queryB}{B_i}%
        =%
        \min_{\pntD \in B_i} \distM{\mtrA}{\queryB}{\pntD}%
        \geq%
        \frac{\LDist{i}}{120}%
        =%
        2 \frac{\LDist{i}}{240}%
        >%
        2 \max_{\pntC \in \modA_i} \distM{\mtrA}{\queryB}{\pntC}.
    \end{align*}
    We conclude that $B_i \cap \modA_i = \emptyset$.
\end{proof}

\begin{remark}%
    \remlab{technicality}%
    A subtle (but minor) technicality is that we require $\distB \neq
    0$, where $\distB = \distM{\mtrA}{\annA}{\annAf}$. This can be
    enforced by replicating every point of $\PntSet$, and assigning
    infinitesimally small positive to the distance between a point and
    its copy. Clearly, for this modified point set this condition
    holds.
\end{remark}

\begin{lemma}%
    \lemlab{qswsp4}%
    Let $\query, \queryB$ be two points in $\QuerySetA_i$, such that
    $\queryB$ was added after $\query$. For a sufficiently small
    $\eps$ and a sufficiently large $c$, we have that
    $\distM{\mtrA}{\query}{\queryB} > \eps \distA_2 / 5$.    
\end{lemma}

\begin{proof}
    We assume for the sake of contradiction that
    $\distM{\mtrA}{\query}{\queryB} \leq \eps \distA_2/5$.  Let
    $\annAB \in A_i$ be the $(1 + \eps/10)$-\ANN found by the
    algorithm for $\queryB$, and let $\annBB$ be the $(1 +
    \eps/10)$-\ANN of $\queryB$ in $\PntSet \setminus
    \ballExt{\mtrA}{\annAB}{\eps \distC/3}$, where $\distC =
    \distM{\mtrA}{\queryB}{\annAB}$.  We have
    \begin{align*}
        \distC = \distM{\mtrA}{\queryB}{\annAB}%
        \leq %
        \pth{1 + \frac{\eps}{10}} \distM{\mtrA}{\queryB}{\PntSet}%
        \leq%
        \pth{1 + \frac{\eps}{10}} \distM{\mtrA}{\queryB}{\annA}%
        \leq %
        \frac{5}{4 \eps} \pth{1 + \frac{\eps}{10}}\distB%
        <%
        \frac{3}{2 \eps} \distB,
    \end{align*}
    by \lemref{qswsp1} (i) and as $\annAB$ is a $(1+\eps/10)$-\ANN of
    $\queryB$ in $\PntSet$. The strict inequality follows under the assumption 
    $\distB > 0$, see \remref{technicality}.  As in \lemref{qswsp3} let 
    $\modA_i = A_i
    \cup \brc{\annAf}$. By \lemref{qswsp3}, we have
    \begin{align*}
        \distM{\mtrA}{\queryB}{B_i} >%
        (1 + \eps/10) \max_{\pntC \in \modA_i}
        \distM{\mtrA}{\queryB}{\pntC},
    \end{align*}
    as $\distM{\mtrA}{\queryB}{B_i} > 2 \max_{\pntC \in \modA_i}
    \distM{\mtrA}{\queryB}{\pntC}$.  If $\modA_i$ is not contained in
    $\ballExt{\mtrA}{\annAB}{\eps \distC/3}$, then there is a point in
    $\modA_i \setminus \ballExt{\mtrA}{\annAB}{\eps \distC/3}$ that
    is, by a factor of $(1+\eps/10)$, closer to $\queryB$ than
    $B_i$. But this implies that $\annBB \notin B_i$, and this is a
    contradiction to the definition of $\queryB$ ($\queryB$ by
    definition has $\annAB \in A_i$ and $\annBB \in B_i$).  Thus,
    $\modA_i$ is contained in $\ballExt{\mtrA}{\annAB}{\eps
       \distC/3}$.

    As such, we have $\annA \in A_i \subseteq \modA_i \subseteq
    \ballExt{\mtrA}{\annAB}{\eps \distC/3}$ (and, by definition
    $\annAf \in \modA_i$, and thus $\annAf$ also belongs to this
    ball). We conclude 
    \begin{align*}
        \distB %
        =%
        \distM{\mtrA}{\annA}{\annAf}%
        \leq%
        \distM{\mtrA}{\annA}{\annAB} + \distM{\mtrA}{\annAB}{\annAf}%
        \leq%
        2 \frac{\eps \distC}{3}%
        < %
        \frac{2\eps }{3}%
        \cdot \frac{3}{2 \eps} \distB = \distB ,
    \end{align*}
    for $\eps$ sufficiently small. This is a contradiction.
\end{proof}

\begin{lemma}%
    \lemlab{qsws}%
    Let $\query, \queryB$ be two points in $\QuerySetA_i$, such that
    $\queryB$ was added after $\query$. Then for sufficiently small
    $\eps$ and sufficiently large $c$ we have,
    $\distM{\mtrB}{\mapped{\query}}{\mapped{\queryB}} =
    \distM{\mtrA}{\query}{\queryB} > \eps \distA_2 / 5 =
    \ordergeq{\eps \LDist{i}}$.
\end{lemma}

\begin{proof}
    Since $\query, \queryB \in \manifold$ it follows from
    \lemref{embed:low:d:d} that
    $\distM{\mtrB}{\mapped{\query}}{\mapped{\queryB}} =
    \distM{\mtrA}{\query}{\queryB}$.  By \lemref{qswsp4}, we have
    $\distM{\mtrA}{\query}{\queryB} > \eps \distA_2 / 5$. From
    \lemref{distA2lb} it follows that $\eps \distA_2 / 5 = \ordergeq
    {\eps \LDist{i}}$.
\end{proof}

\begin{lemma}%
    \lemlab{qsub}%
    We have that $\max \pth{\cardin{\QuerySetA_i},
       \cardin{\QuerySetB_i}} = \eps^{-\porder{\dd}}$.
\end{lemma}

\begin{proof}
    From \lemref{diamub} and \lemref{qsws} it
    follows that the spread of the set $\mapped{\QuerySetA_i}$ is
    bounded by
    \begin{equation*}
        \order { \frac{\LDist{i}/\eps}{\eps \LDist{i}}}%
        =%
        \order{\frac{1}{ \eps^2}}.
    \end{equation*}
    
    Since $\mapped{\QuerySetA_i} \subseteq \mtrB$ which is a space of doubling
    dimension $\porder{\dd}$ it follows that $\cardin{\mapped{\QuerySetA_i}} =
    \eps^{-\porder{\dd}}$. The same argument works for $\mapped{\QuerySetB_i}$.
    For any $\query \in \manifold, \mapped{\query} = (\query, 0)$ and it is easy
    to see that the mapping $\query \to \mapped{\query}$ is bijective. As such
    $\cardin{\mapped{\QuerySetA_i}} = \cardin{\QuerySetA_i}$,
    and similarly $\cardin{\mapped{\QuerySetB_i}} = \cardin{\QuerySetB_i}$, and
    the claimed bounds follow.
\end{proof}

The next lemma bounds the number of regions created.
\begin{lemma}
    The number of regions created by the algorithm is $n/
    \eps^{\porder{\dd}}$.
\end{lemma}

\begin{proof}
    As shown in \lemref{ftrub} the number of outer regions created is
    bounded by $\eps^{-\porder{\dd}}$. Consider an inner region
    $\RegionA_{\query}$. For this point $\query$ the algorithm found a
    valid $\annA$ and $\algoonnB$. Now from the definition of a \WSPD
    there is some $i$ such that $\mapped{\annA} \in \mapped{A_i},
    \mapped{\algoonnB} \in \mapped{B_i}$ or $\mapped{\annA} \in
    \mapped{B_i}, \mapped{\algoonnB} \in \mapped{A_i}$.  In other
    words there is some $i$ such that $\query \in \QuerySetA_i$ or
    $\query \in \QuerySetB_i$. As shown in \lemref{qsub} the size of
    each of these is bounded by $\eps^{-\porder{\dd}}$. Since the
    total number of such sets is $2m$ where $m = n c^{\porder{\dd}}$
    is the number of pairs of the \WSPD, it follows that the total
    number of inner regions created is bounded by
    $\pth[]{c/\eps}^{\porder{\dd}}n \leq n \eps^{-\porder{\dd}}$, for
    $\eps$ sufficiently small.
\end{proof}

\subsection{The result}

We summarize the result of this section.

\begin{theorem}
    The online algorithm presented in \figref{online-ann} always
    returns a $(1 + \eps)$-\ANN. If the query points are constrained
    to lie on a subspace of doubling dimension $\dd$, then the maximum
    number of regions created for the online \AVD by the algorithm
    throughout its execution is $n / \eps^{\porder{\dd}}$.
\end{theorem}

% ----------------------------------------------------------------------
% ----------------------------------------------------------------------
% ----------------------------------------------------------------------

\section{Conclusions}
\seclab{conclusions}

In this paper, we considered the \ANN problem when the data points can
come from an arbitrary metric space (not necessarily an Euclidean
space) but the query points come from a subspace of low doubling
dimension. We demonstrate that this problem is inherently low
dimensional by providing fast \ANN data-structures obtained by
combining and extending ideas that were previously used to solve \ANN
for spaces with low doubling dimensions.

Interestingly, one can extend Assouad's type embedding to an embedding
that $(1+\eps)$-preserves distances from $\PntSet$ to $\manifold$ (see
\cite{hm-fcnld-06} for an example of a similar embedding into the
$\ell_\infty$ norm). This extension requires some work and is not
completely obvious.  The target dimension is roughly $1/\eps^{O(\dd)}$
in this case. If one restricts oneself to the case where both
$\PntSet$ and $\manifold$ are in Euclidean space, then it seems one
should be able to extend the embedding of Gottlieb and Krauthgamer
\cite{gk-nadr-11} to get a similar result, with the target dimension
having only polynomial dependency on $\dd$. However, computing either
embedding efficiently seems quite challenging. Furthermore, even if
the embedded points are given, the target dimension in both cases is
quite large, and yields results that are significantly weaker than the
ones presented here.

% SARIEL - new stuff

The on the fly construction of \AVD without any knowledge of the query
subspace (\secref{oann}) seems like a natural candidate for a
practical algorithm for \ANN. Such an implementation would require an
efficient way to perform point-location in the generated regions. We
leave the problem of developing such a data-structure as an open
question for further research. In particular, there might be a middle
ground between our two \ANN data-structures that yields an efficient
and practical \ANN data-structure while having very limited access to
the query subspace.

% ----------------------------------------------------------------------
% ----------------------------------------------------------------------

\bibliographystyle{alpha}%
\bibliography{low_dim_ann}%
%\bibliography{shortcuts,geometry}%

\end{document}